 \PassOptionsToPackage{varbb}{newtx}
\documentclass[acmsmall,authorversion,dvipsnames]{acmart}

\newcommand{\arxivorappendix}[0]{appendix}

\usepackage[ruled,vlined,linesnumbered]{algorithm2e} % noend
\DontPrintSemicolon
\SetKwInOut{Input}{input}
\SetKwInOut{Output}{output}
\SetKwProg{Fn}{Function}{:}{}
\SetKw{Accept}{accept}
\SetKw{Reject}{reject}
\SetKw{ACCEPT}{accept}
\SetKw{REJECT}{reject}

\SetCommentSty{commentfont}

\usepackage{mathtools}
\usepackage{enumitem}
\usepackage{subcaption}
\usepackage{soul}
\usepackage{csquotes}
\usepackage{tikz}
\usetikzlibrary{snakes,arrows,shapes,automata,positioning}
\usetikzlibrary{calc}
\usetikzlibrary{topaths}
\usetikzlibrary{backgrounds}
\usetikzlibrary{graphs}

\usepackage{xassoccnt}
\usepackage{dashbox}
\usepackage[capitalise]{cleveref}
\usepackage{mathtools}

\tikzset{vertex/.style={minimum width=4pt,inner sep=0pt,circle,fill=black}}
\tikzset{hyperedge/.style={thick}}
\tikzset{hypergraph/.style={every label/.style={font=\scriptsize,circle,inner sep=1pt}}}
\tikzset{subtle/.style={on background layer,every path/.append style={lightgray!50!white}}}

\input{todonotes}

\AtBeginDocument{%
  }

\acmJournal{PACMMOD}
\acmYear{2026} \acmVolume{4}
\acmNumber{2 (PODS)} \acmArticle{97}
\acmMonth{5} % \acmPrice{}
\acmDOI{10.1145/3801893}
\received{December 2025}
\received[accepted]{February 2025}

\DeclareMathOperator*{\ew}{{\mathsf{ew}}}				    % Edge-cover width
\DeclareMathOperator*{\fchw}{{\mathsf{fcghw}}}			    % Free-connex generalized hypertree width

\DeclareMathOperator*{\vars}{{\mathsf{vars}}}			    % Set of variables

\newcommand*{\SP}{\ensuremath{\mathsf{\#P}}}                % #P
\newcommand*{\numad}{\ensuremath{\mathsf{\#ADORNMENTS}}}    % #ADORNMENTS problem
\newcommand*{\pptcnf}{\ensuremath{\mathsf{\#PP2CNF}}}       % #PP2CNF problem

\newcommand*{\TC}{\ensuremath{\mathsf{TC}}}                 % TC
\newcommand*{\HT}{\mathcal T}               

\newcommand*{\ar}{\ensuremath{\mathsf{ar}}}                 % arity
\newcommand*{\ear}{\ensuremath{\mathsf{ear}}}               % EDB arity
\newcommand*{\NEDBs}{\ensuremath{\mathsf{\#EDBs}}}          % Number of EDB relations
\newcommand*{\NIDBs}{\ensuremath{\mathsf{\#IDBs}}}          % Number of IDB relations

\DeclarePairedDelimiter{\set}{\lbrace}{\rbrace}

\newcommand*{\vertices}[2][]{V#1(#2#1)}

\newcommand{\var}{\vars}

\newcommand{\sem}[1]{{\llbracket{}{#1}\rrbracket}}

\newcommand*{\galgo}{\ensuremath{g_{\sf out}}}
\newcommand*{\gmcd}{\ensuremath{g_{\sf min}}}

\begin{document}

\title{Size Bound–Adorned Datalog}


\author{Christian Fattebert}
\affiliation{%
  \institution{EPFL}
  \city{Lausanne}
  \state{Vaud}
  \country{Switzerland}
}
\email{christian.fattebert@epfl.ch}
\orcid{0009-0000-5792-2426}
\author{Zhekai Jiang}
\affiliation{%
  \institution{EPFL}
  \city{Lausanne}
  \state{Vaud}
  \country{Switzerland}
}
\email{zhekai.jiang@epfl.ch}
\orcid{0009-0001-0989-5926}
\author{Christoph Koch}
\affiliation{%
  \institution{EPFL}
  \city{Lausanne}
  \state{Vaud}
  \country{Switzerland}
}
\email{christoph.koch@epfl.ch}
\orcid{0000-0002-9130-7205}
\author{Reinhard Pichler}
\affiliation{%
  \institution{TU Wien}
  \city{Vienna}
  \country{Austria}
}
\email{pichler@dbai.tuwien.ac.at}
\orcid{0000-0002-1760-122X}
\author{Qichen Wang}
\affiliation{%
  \institution{Nanyang Technological University}
  \city{Singapore}
  \country{Singapore}
}
\email{qichen.wang@ntu.edu.sg}
\orcid{0000-0002-0959-5536}

% \renewcommand{\shortauthors}{}

\begin{abstract}
%With a growing demand for datalog engines and ever-increasing data sizes, cost-based optimization has become essential for efficiently handling large amounts of data in such engines.
% A core requirement for CBO is the accurate estimation of intermediate relation sizes.
%While there are numerous methods for non-recursive conjunctive queries, no reliable technique exists for computing size bounds of intermediate results and reasoning about complexities of recursive queries.
We introduce EDB-bounded datalog, a framework for deriving upper bounds on intermediate result sizes and the asymptotic complexity of recursive queries in datalog.  We present an algorithm that, given an arbitrary datalog program, constructs an EDB-bounded datalog program in which every rule is adorned with a (non-recursive) conjunctive query that subsumes the result of the rule, thus acting as an upper bound. From such adornments, we define a notion of width based on (integral or fractional) edge-cover widths. Through the adornments and the width measure, we obtain, for every IDB predicate, worst-case upper bounds on their sizes, which are polynomial in the input data size, given a fixed program structure. Furthermore, with these size bounds, we also derive fixed-parameter tractable, output-sensitive asymptotic complexity bounds for evaluating the entire program. Additionally, by adapting our framework, we obtain a semi-decision procedure for datalog boundedness that efficiently rewrites most practical bounded programs into non-recursive equivalent programs.
%, offering a path towards more efficient execution.
\end{abstract}




\maketitle

\section{Introduction}
\label{sec:introduction}

% Motivation:  People need fast datalog engines that can handle large amounts of data.  Cost-based optimization is necessary, and cost-based optimization requires us to be able to estimate the sizes of intermediate relations, but so far nobody knows how to compute size bounds on recursively defined relations.  This paper aims to improve this situation.

% Datalog: chaudhuri_souffle_2016 – Soufflé,  swift_xsb_2012 – XSB
% "Datalog-like": leone_dlv_2006 – DLV,  aref_design_2015 – LogicBlox,  aref_rel_2025 – RelationalAI,  gopalakrishnan_z_2011 – the Datalog interface in Z3,  zhang_better_2023 – Egglog,  marussy_refinery_2024 – Refinery
Datalog engines~\cite{chaudhuri_souffle_2016,swift_xsb_2012,leone_dlv_2006,aref_design_2015,aref_rel_2025,gopalakrishnan_z_2011,zhang_better_2023}
are starting to show their potential in a wide range of data-intensive applications,
from static program analysis~\cite{shi_datalog-based_2025,wu_diffbase_2021,de_moor_using_2011, scholz_fast_2016,noauthor_ctadl_nodate}
to large-scale graph processing~\cite{aref_rel_2025,marussy_refinery_2024,bellomarini_vadalog_2018,seo_distributed_2013,shkapsky_graph_2013}.
With ever-growing data sizes, cost-based optimization is essential for fast, scalable execution in such engines.
However, accurate cost-based optimization relies heavily on the ability to estimate the cardinalities of intermediate results generated during query evaluation~\cite{leis_how_2015}.
In traditional join processing, there have been significant research efforts on
cardinality estimation~\cite{zhang_lpbound_2025,abo_khamis_join_2024,atserias_size_2013,gottlob_size_2012,cai_pessimistic_2019},
cost-based optimization~\cite{selinger_access_1979,moerkotte_dynamic_2008,stoian_dpconv_2024}, and
theoretical complexity bounds~\cite{yannakakis_algorithms_1981,gottlob_hypertree_2002,wang_yannakakis_2025,hu_output-optimal_2025,deep_output-sensitive_2024}.
But estimating result sizes is significantly more challenging when recursion is involved, which is an inherent feature of datalog. This is also the case for recursive SQL~\cite{international_organization_for_standardization_iso_information_1999}, which is starting to gain popularity~\cite{burzanska_recursive_2024,schule_recursive_mining_2022,schule_recursive_2022}.
In particular, there is currently no reliable method for computing size bounds on the final outputs of recursively-defined relations, except for a few very specific datalog programs~\cite{seshadri_expected_1995}.
Without such bounds, optimizers are often forced to rely on (1) heuristics, as is the case for the popular datalog engine Souffl\'e~\cite{chaudhuri_souffle_2016,villanueva_building_2022}, or (2) runtime optimization techniques~\cite{herlihy_adaptive_2024,herlihy_static_2025}.
Such methods often lead to suboptimal execution plans or higher runtime overhead for complex workloads.

The present paper aims to improve this situation by developing techniques for computing (asymptotic) bounds on the result sizes of recursive datalog programs.
To this end, we introduce a framework for adorning datalog rules (and specifically IDB atoms) with conjunctive queries over the input database, the union of which subsumes the result of the IDB predicate, thus acting as a form of upper bound. We refer to such adorned programs as \emph{EDB-bounded datalog programs}.
We give a simple illustration with a variant of the well-known transitive closure program:

\begin{example}\label{ex:tc-intro}
    Consider the datalog program $\Pi$ with the following rules:
    \begin{align*}
        R_1 : \mathsf{TC}(\bar x, \bar y) \gets 
        &\ 
        e(\bar x, \bar y),
        \\
        % \quad\quad
        R_2 : \mathsf{TC}(\bar x, \bar y) \gets 
        &\ 
        \mathsf{TC}(\bar x, \bar z), e(\bar z, \bar y),
    \end{align*}
    where $e$ is an EDB relation, each of $\bar x$, $\bar y$, and $\bar z$ is a vector of $m$ variables, and $e$ and $\TC$ hence have arity $2m$.
    % Given a set of edges in $e$, this program computes all pairs of $(x, y)$ such that $y$ is reachable from $x$.
    This program is \emph{unbounded}, as the depth of recursion depends on the size of the input.
    Therefore, there does not exist a finite non-recursive program (or, finite union of conjunctive queries~(UCQ)) that is equivalent to $\Pi$.
    However, assuming that the number of tuples in $e$ is $N$, we can in fact see that the number of tuples in the final output of $\mathsf{TC}$ is bounded by $N^2$, even though the arity $2m$ may be greater than 2.
 %   , which is achieved when the edges form a cycle.
    This is because, at maximum, we choose (1) one tuple from $e$ and assign the first $m$ attributes to $\bar x$, and (2) one tuple from $e$ and assign the last $m$ attributes to $\bar y$.
    
    In this paper, we generalize this idea by deriving size bounds based on a notion of width analogous to edge cover numbers in conjunctive queries.
    For example, following our method, the program $\Pi$ above would be rewritten to the following program $\Pi'$:
    \begin{alignat*}{3}
        R_1' : &\quad \TC_{\TC(\bar x, \bar y) \gets e(\bar x, \bar y)}(\bar x, \bar y) && \gets  e(\bar x, \bar y), \\
        R_2' : &\quad \TC_{\TC(\bar x, \bar y) \gets e(\bar x, \_), e(\_, \bar y)}(\bar x, \bar y) && \gets \TC_{\TC(\bar x, \bar y) \gets e(\bar x, \bar y)}(\bar x, \bar z),{}  & e(\bar z, \bar y), \\
        R_3' : &\quad \TC_{\TC(\bar x, \bar y) \gets e(\bar x, \_), e(\_, \bar y)}(\bar x, \bar y) && \gets \TC_{\TC(\bar x, \bar y) \gets e(\bar x, \_), e(\_, \bar y)}(\bar x, \bar z),{} & e(\bar z, \bar y).
    \end{alignat*}
    In $\Pi'$, each atom of the IDB predicate $\TC$ is adorned by a conjunctive query which shows the ``provenance'' of the tuples involved.
    More importantly, the rule in the adornment of the head atom of each rule subsumes the results that will be produced by the rule.
    Therefore, $\Pi'$ is equivalent to $\Pi$ in the sense that the tuples that will be in the final output of the IDB relation $\TC$ in the original program $\Pi$ is exactly the same as the union of the two adorned variants of $\TC$ in $\Pi'$, namely $ \TC_{\TC(\bar x, \bar y) \gets e(\bar x, \bar y)}$ and $\TC_{\TC(\bar x, \bar y) \gets e(\bar x, \_), e(\_, \bar y)}$.
    In such a case, we say the edge-cover width of $\TC$ is 2, and we can then easily tell that the size of $\TC$ is bounded by $N^2$.
    With such bounds on the output sizes, we can further derive output-sensitive, fixed-parameter tractable asymptotic complexity bounds for evaluating the entire program, which, in this case of $\TC$, is $O(f(\Pi) \cdot N^2)$, where $f$ is a function that depends only on the structure of the program $\Pi$.
\end{example}

%While this particular example is relatively simple, there is currently no generic way to perform such reasoning for arbitrary datalog programs, despite some results on a few specific ones~\cite{seshadri_expected_1995}.  Unlike previous work on datalog size bounds~\cite{seshadri_expected_1995} which considers the distribution of the input database instance, our approach relies only on the structural properties of the program.

Our approach is thus based on rewriting---or, more precisely, adorning---datalog programs with non-recursive queries that serve as upper bounds on the IDB relations defined by the program. While a datalog engine might use such bounding queries for cardinality estimation, in this paper, we focus on proving asymptotic bounds based on structural properties of the datalog programs---specifically, integral and fractional edge-cover width.

As a central result, we
show that our program rewriting technique can be used to find the minimal integral edge-cover width (which immediately yields an asymptotic bound on output sizes) across all datalog programs equivalent to a given input program. In other words, determining the worst-case sizes of the IDB relations of a datalog program, or any datalog program equivalent to it, is decidable.
This may appear somewhat surprising, as semantic properties such as (program) boundedness or equivalence of recursive datalog programs are generally undecidable~\cite{DBLP:journals/jlp/HillebrandKMV95,DBLP:journals/jlp/Shmueli93}.

%This is achieved without having to explicitly consider these equivalents.

We demonstrate the tightness of integral edge-cover widths and our size bounds by constructing database instances where they are reached.

We show that output sizes of IDB relations of a program $\Pi$ have worst-case upper bounds of the form $f(\Pi) \cdot N^{w}$, which is polynomial in the input data size $N$ for a fixed structure of $\Pi$.

Based on our size bounding machinery, we obtain corresponding output-sensitive, fixed-parameter tractable complexity bounds for evaluating entire programs. These results are based on insights into incrementally evaluating our rewritten programs.  We present results on general datalog as well as tighter bounds on a number of datalog fragments---linear, simple chain, and adornment-groundable programs (a newly defined class).

Our adornment rewriting algorithm is fundamentally a bottom-up fixpoint procedure that is guaranteed to terminate given suitable methods for comparing and relaxing (simplifying) adornments. A variation of this procedure unrolls recursive datalog programs into (possibly infinite) UCQs, and terminates if the program is bounded---a semi-decision procedure for datalog boundedness.  We show in this paper how our adornment algorithm can be instrumented to search for non-recursive rewritings (for bounded programs) with a memory budget, gracefully degrading to our EDB-bounded datalog programs (as in \cref{ex:tc-intro}) when the input program is not bounded or cannot be proven to be within the memory budget.

\smallskip

The structure of the paper is as follows:
\begin{itemize}[leftmargin=*]
    \item We start by laying out the necessary background in \cref{sec:preliminaries}.
    \item In \cref{sec:edbAdornedPrograms}, we present our technique and algorithm for adorning IDB atoms in a datalog program with conjunctive queries, and show its correctness (conservativity). 
    We present this approach as a generic framework that may be adapted for different use cases.

    \item In \cref{sec:boundedness}, we revisit the problem of datalog boundedness. We show its semi-decidability, and give one instantiation of our framework to either (1) rewrite most practical bounded programs into equivalent UCQs by inlining, or (2) give an equivalent EDB-bounded program that is recursive but still allows for reasoning about size bounds and complexity.
    \item In Section~\ref{sec:UpperBoundIdbs}, we obtain size bounds on IDB relations via a notion of edge-cover widths based on the adornments in EDB-bounded programs.  \item We discuss the optimization and minimization of our program rewritings in Section~\ref{sec:minimize-edb-bounded}.
    
    \item \cref{sec:complexity-bounds} presents our fixed-parameter tractable complexity bounds for general datalog as well as tighter results for specific fragments of datalog.
\end{itemize}

Due to space constraints, all omitted proofs in the rest of the paper can be found in the \arxivorappendix.

\section{Preliminaries}
\label{sec:preliminaries}

\paragraph{Datalog}
In this paper, we consider standard positive datalog programs and adopt the standard naming conventions.
A datalog program $\Pi$ consists of a set of \emph{rules} of the form
$$ q(\bar s) \gets p_1(\bar s_1), \dots, p_\mu(\bar s_\mu), \quad \text{or, equivalently,} \quad q(\bar s) \gets \bigwedge_{i = 1}^{\mu} p_i(\bar s_i), $$
where $q(\bar s)$ is called the \emph{head atom} and all $p_i(\bar s_i)$ are \emph{body atoms}.
% For convenience, we may alternatively write the bodies of rules as conjunctions in the following form:
% $$ q(\bar s) \gets \bigwedge_{i = 1}^{\mu} p_i(\bar s_i). $$
Each atom $p(\bar t)$ consists of a predicate symbol $p$ and a tuple of terms $\bar t = (t_1, \dots, t_\alpha)$.
We use the same set of symbols, such as $p, q, \dots$, to denote both a predicate symbol and the relation corresponding to that predicate.
A predicate $p$ belongs to the intensional database (IDB) if it is in the head of a rule with a non-empty body. Otherwise, it belongs to the extensional database (EDB).
Without loss of generality, we assume that each rule is ordered so that the first $\lambda$ atoms (for some $0 \leq \lambda \leq \mu$) have IDB predicates and the rest have EDB predicates.
Within a tuple $\bar t$ of an atom $p(\bar t)$, each term $t_i$ is either a variable or a constant. We write $\var (\bar t)$ to denote the set of variables occurring in $\bar t$.
% By slight abuse of notation, we will sometimes treat such a tuple $\bar t$ as a set by writing e.g., $x \in \bar t$, meaning that  $x$ is one of the components of the tuple $\bar t$.
The wildcard ``\_'' is employed to denote a freshly introduced, existentially quantified variable.
The number of terms $\alpha$ is called the \emph{arity} of the predicate $p$, generally denoted by $\ar(p)$.
A datalog rule is defined as \emph{safe} if every variable present in the head also appears at least once in the body.
The notations $\NEDBs$ and $\NIDBs$ represent the total number of EDB predicates and IDB predicates, respectively, within a datalog program. 
Given a datalog program $\Pi$, its number of rules is denoted by $|\Pi|$.

For an IDB predicate $q$ in $\Pi$ and a database instance $D$, we write $\sem{q(\Pi,D)}$ to denote the 
IDB relation $q$ that can be derived with the program $\Pi$ from the database instance $D$.
That is, $\sem{q(\Pi,D)}$ contains all tuples $\bar s$, such that the fact $q(\bar s)$ 
is contained in the least fixpoint of the immediate consequence operator defined by 
the program $\Pi$, when starting with the EDB $D$.

We say a datalog rule $R_0$ is \emph{subsumed} by a rule $R_1$, if there is a homomorphism $h$ from $R_1$ to $R_0$, such that $h$ restricted to the head variables is a variable renaming. That is, all the tuples produced by rule $R_0$ are also produced by rule $R_1$.

We use \emph{inlining} to refer to the process of replacing an IDB atom $q(\bar t)$ in some rule $R$ with the body of a rule $R': q(\bar s) \gets B_1, \dots, B_n$ that has $q$ as the head predicate, under a substitution $\sigma$ such that $\sigma(\bar s) = \bar t$, which yields instantiated body atoms $\sigma(B_1),\dots \sigma(B_n)$, 
where $\sigma$ must not rename any variable that is not in the head $\bar s$ to the name of some variable in the original rule $R$.
For instance, in the following program
\begin{alignat*}{3}
    & R_1 : &\ p_0(x, y) \gets &\ e(x, z), e(z, y),\\
    & R_2 : &\ q(x, y) \gets &\  p_0(x, y), p_1(z),
\end{alignat*}
we can obtain $q(x, y)\gets e(x, z_0), e(z_0, y), p_1(z)$ by one step of inlining. Note that, while doing so, $z$ in $R_1$ is renamed to $z_0$ to avoid collision with the variable $z$ in the original rule $R_2$.

\paragraph{Unification}
A \emph{unifier} of two tuples $\bar s$ and $\bar t$ of terms is a substitution $\sigma$ such that, when applied to the tuples, $\sigma$ makes them syntactically equal, i.e., $\sigma(\bar s) = \sigma(\bar t)$.
When such a unifier exists, we say that the two tuples unify.
Given a set of pairs of tuples $\{ (\bar s_i, \bar t_i) \mid 1 \leq i \leq \mu \}$, a \emph{simultaneous unifier} is a unifier such that $\sigma(\bar s_i) = \sigma(\bar t_i)$ for all $i$.
A unifier is called a \emph{most general unifier} (mgu) if every other unifier can be obtained from it by further instantiating variables.
mgus can be computed in time linear in the size of the input set~\cite{DBLP:journals/toplas/MartelliM82}.
For instance, $\sigma=\{x\mapsto x, y\mapsto y,z\mapsto z, w\mapsto z\}$ is a most general unifier for the tuples $(x, y, z, z)$ and $(x,y,z,w)$.

\paragraph{Proof trees}
Given a datalog program $\Pi$ and an EDB $D$, a \emph{proof tree} for a ground atom $A$ in the IDB is a finite, rooted tree whose root is labeled by $A$, whose leaves are EDB atoms in $D$, and whose edges correspond to rule applications. More specifically, 
if a node is labeled by an atom $A$ and its children are labeled $B_1, \dots, B_\mu$, then 
the edge from $A$ to each $B_i$ is labeled by a rule $R$ of the form 
$A' \gets B'_1, \dots, B'_\mu$, such that $A = \sigma(A')$  and 
$B_i = \sigma(B'_i)$ for some \emph{ground substitution} $\sigma$, i.e., a substitution that replaces every variable with a constant.

\paragraph{Derived atoms}
Given a rule $R$ with head predicate $q$ in a datalog program $\Pi$, and given an EDB $D$, 
we say that a ground atom $q(\bar t)$ is derived by rule $R$ over $D$ if there exists some proof tree whose root is labeled by $q(\bar t)$ and such that the edges from the root are labeled by the rule $R$.

%%%%%%%%%%%%%%%%%%%%%%%%%%%%%%%%%%%%%%%%%%%%%%%%%%%%%%%%%%%%%%%%%%%%

\section{EDB-Adorned and EDB-Bounded Datalog Programs}
\label{sec:edbAdornedPrograms}

In this section, we begin to introduce our framework to provide bounds for datalog programs via adornments added to IDB atoms. 
These adornments themselves are datalog rules with EDB atoms only, and they subsume the adorned IDB predicate.
We formally define the related notations and give a generic algorithm to construct such adornments.
The algorithm is parameterized and customizable for different use cases, as we will show in the following sections.
% In particular, we first show that our framework can be thought of as inlining or unrolling of datalog program, which leads to a procedure to show semi-decidability of the boundedness of datalog programs.
% By adjusting the algorithm, we construct an equivalent datalog program that helps to reason about the sizes of IDB predicates and the complexity of evaluation, as we will elaborate in later sections.

\subsection{Basic Definitions}

We start from the most basic building blocks of our framework.

\begin{definition}[EDB-adorned atom]
\label{def:adorned-atom}
An {\em EDB-adorned atom} is an atom of the form $q_\rho (\bar t)$, such that 
$q$ is an IDB predicate of the given schema of arity $\alpha$, 
$\bar t = (t_1, \dots, t_\alpha)$ is a tuple of the same arity, such that 
each $t_i$ is either a variable or a constant, 
and the \emph{adornment} $\rho$ is a safe datalog rule $A \leftarrow B_1, \dots, B_\mu$
satisfying the following properties: 
\begin{enumerate}[leftmargin=*]
    \item $A$ is an atom with predicate symbol $q$, 
    \item  every body atom $B_i$ has as an EDB predicate and each 
    argument in $B_i$ is either a variable or 
    the underscore ``$\_$''.
    \item $q(\bar t)$ is an instance of $A$---that is, there exists a substitution $\sigma$ such that $\sigma(A)=q(\bar t)$.
\end{enumerate}  
\end{definition}

\begin{definition}[EDB-adorned datalog rule]
\label{def:adorned-rule}
    An \emph{EDB-adorned datalog rule} $R$ is a datalog rule 
    in which the head and every body atom is either an EDB atom or an EDB-adorned IDB atom.
A datalog program with only EDB-adorned rules is called an 
{\em EDB-adorned datalog program}. 
\end{definition}

The {\em intended meaning} of the  adornment $\rho$ of an EDB-adorned IDB atom 
$q_\rho(\bar t)$ is to 
indicate 
% which EDB relations can possibly  provide values for the variables in $\bar t$. Intuitively, $\rho$ tells us, 
how instances of the 
atom $q_\rho (\bar t)$ could be directly derived from EDB atoms via $\rho$.
The rule $\rho$ thus corresponds to iterated inlining, and
the EDB atoms in the rule body of an adornment
constitute restrictions on the tuples that can possibly be derived 
by the adorned IDB predicate.  
Of course, such a restriction is useless in the case of an EDB atom.
We thus only consider adornments for IDB predicates.
We formalize this idea by the following definitions of EDB-bounded rules and programs.

\begin{definition}[EDB-bounded rule]
    Given an EDB-adorned datalog program $\Pi$, let $R$ be a rule in $\Pi$ with head predicate $q_\rho$.
    We say that rule $R$ {\em is bounded}
    by the adornment $\rho$, if for every EDB $D$, the following property holds: 
    If an atom $q_\rho(\bar t)$ is derived by rule $R$ when evaluating $\Pi$ over $D$, 
    then
    $q(\bar t)$ is also derived when evaluating the single-rule program consisting of the rule $\rho$
    over $D$.
\end{definition}

\begin{definition}[EDB-bounded program]
    Given an EDB-adorned datalog program $\Pi$, we call $\Pi$ an {\em EDB-bounded} program if every rule of $\Pi$ is bounded by the adornment of its head predicate.
\end{definition}

To give an intuition, we illustrate again with the transitive closure program $\Pi$ in \cref{ex:tc-intro}.

\begin{example}
\label{ex:TC-def}
    Consider the transitive closure program $\Pi$ in \cref{ex:tc-intro} again.
    For simplicity, from this point, we assume $\bar x$, $\bar y$, and $\bar z$ each contain one term only and hence drop the bars.
    The following datalog program $\Pi''$ is a valid EDB-bounded program:
    \begin{alignat*}{3}
        R_1'' : &\quad \TC_{\TC(x, y) \gets e(x, y)}(x, y) && \gets e(x, y), \\
        R_2'' : &\quad \TC_{\TC(x, y) \gets e(x, z), e(z, y)}(x, y) && \gets \TC_{\TC(x, y) \gets e(x, y)}(x, z), e(z, y).
    \end{alignat*}
    One can easily verify that the adornment of $R_1''$, $\TC(x, y) \gets e(x, y)$, subsumes the rule $R_1''$---in fact, they are equivalent.
    The predicate $\TC_{\TC(x, y) \gets e(x, y)}$ contains only tuples that are produced by $R_1''$, i.e., the pairs $(x, y)$ that are connected via some edge in $e$.
    Similarly, the adornment of $R_2''$ subsumes, and is equivalent to, the rule $R_2''$ as well, if we replace the atom $\TC_{\TC(x, y) \gets e(x, y)}(x, z)$ by $e(x, y)$.
    Notice that the adornments in this case simply correspond to inlinings of the original program. We could do one step further and add the following rule
    $$R_3'' : \TC_{\TC(x, y) \gets e(x, z'), e(z', z), e(z, y)}(x, y) \gets \TC_{\TC(x, y) \gets e(x, z), e(z, y)}(x, z), e(z, y).$$
    Again, this is a valid EDB-bounded rule.
    However, these adorned versions of the predicate $\TC$ in $\Pi''$ only produce pairs of $x$ and $y$ with distance at most 3, but the original program $\Pi$ is unbounded.
    Hence, $\Pi''$ is not equivalent to $\Pi$ (unless we iterate this inlining-style rewriting of rules infinitely). 
    However, in the next section, we will introduce means to get meaningful, finite adorned programs that allow us to bound the 
    IDB predicates of the original program. 
\end{example}

\subsection{Construction of EDB-Bounded Datalog Programs}

In \cref{ex:TC-def}, we have seen an intuition of EDB adornments by viewing them as possible inlinings of the original program and how this may become infinite if we want to bound all tuples produced by the original program. 
It is easy to verify that program $\Pi'$ in \cref{ex:tc-intro} is also an 
EDB-bounded program. In contrast to program $\Pi''$, the adornments of the predicate $\TC$ in $\Pi'$ 
actually do allow us to bound the IDB relation $\TC$ of the original program $\Pi$: Intuitively, the
adornments of $\TC$ in $\Pi'$ capture the observation that a tuple $(x, y)$ generated by $\TC$ in program $\Pi$ 
could either be (1) two endpoints of the same edge, hence $e(x, y)$, or (2) the source of one edge and target of another, hence $e(x, \_), e(\_, y)$. 
Alternatively, we can consider these adornments as taking all possible inlinings, but ``relaxing'' the rules 
by replacing variables with ``$\_$'' or dropping body atoms. 
We generalize this idea in the following definition of \emph{adornment relaxation functions}.

\begin{definition}[adornment relaxation function]
\label{def:relax}
    An \emph{adornment relaxation function} $g$ is one that takes as input an adornment
    $ \rho : p(\bar t) \gets e_1(\bar t_1), \dots, e_\mu(\bar t_\mu), $
    where all body atoms $e_1(\bar t_1), \dots, e_\mu(\bar t_\mu)$ are EDB atoms,
    and performs a sequence of operations, each of which either removes a body atom or replaces an attribute in a body atom by ``$\_$'', such that $g(\rho)$ remains a safe datalog rule.
\end{definition}

Clearly, if a rule $g(\rho)$ results from ``relaxing'' a rule $\rho$, then 
$g(\rho)$ subsumes $\rho$:

\begin{proposition}
    For any adornment relaxation function $g$ and safe datalog rule $\rho$ with only EDB atoms in the body, $g(\rho)$ subsumes $\rho$.
\end{proposition}

\begin{example}
\label{ex:relax-fns}
We show a few examples of adornment relaxation functions:
\begin{itemize}[leftmargin=*]
    \item $\mathit{id}$: The identity function, which does not modify the rule at all.
    \item $\galgo$:
    For each body atom $e(\bar t)$ and for each attribute $t_i$ in $\bar t$, 
    we replace $t_i$ by $\_$ if $t_i$ does not occur in the head of the rule (i.e., if $t_i$ is not in the ``output'').
    Finally, we remove all body atoms $e(\bar t)$
    that impose less restrictions than some other body atom $e(\bar t')$. That is, 
    for each $i = 1, \dots, \ar(e)$, either $t_i = t'_i$ or $t_i = \_$.
    Note that, in particular, this removes all body atoms $e(\bar t)$ that have only ``$\_$''s remaining in $\bar t$. \cref{ex:propagation-falgo} illustrates the application of this function, and this function will be used to derive edge-cover widths and size bounds in Sections \ref{sec:UpperBoundIdbs} and \ref{sec:complexity-bounds}.
    % DONT REMOVE g_k and g_min as they are used for boundedness and minimization
    \item $g_{k}$: A function that acts as the identity function when the input rule contains fewer than $k$ atoms and applies $\galgo$ otherwise. This function will be used in \cref{sec:boundedness} to try to rewrite a bounded program into a non-recursive equivalent version with some storage budget.
    \item $\gmcd$:
    This function retains only a minimal subset of body atoms from the original rule such that all head variables are contained within the set. Then, similar to $\galgo$, it replaces all variables that are not in the head by ``$\_$''s. Furthermore, any duplicate variable occurrences in the body are replaced by ``$\_$'', so that each head variable appears only once in the resulting rule. This function will be used in \cref{sec:minimize-edb-bounded} to reduce the number of rules in the transformed program.
\end{itemize}
\end{example}

Intuitively, different adornment relaxation functions have different levels of relaxation.
Moreover, in combination with the inlining-style adornments in Example~\ref{ex:TC-def},
this notion of adornment relaxation function suggests a way in which adornments can be propagated algorithmically.

\begin{example}
    \label{ex:propagation-falgo}
Consider the transitive closure program $\Pi$ from Example~\ref{ex:tc-intro} and suppose that we generate adorned rules as in Example~\ref{ex:TC-def}.
But now we apply 
the adornment relaxation function $\galgo$ to every rule thus generated.
We initialize $\Pi' = \emptyset$.
The first rule $R_1$ in $\Pi$ gives rise to the following EDB-adorned rule which we add to $\Pi'$:
    \[ R_1' : \mathsf{TC}_{\mathsf{TC}(x, y) \gets e(x, y)}(x, y) \gets e(x, y). \]

Next, consider rule $R_2$ in $\Pi$,  but replace the IDB atom in the body of $R_2$ by the adorned head atom of the rule $R'_1 \in \Pi'$. Moreover, to get the adorned rule $R'_2$,  we apply $\galgo$ to the adornment $\mathsf{TC}(x, y) \gets e(x, z), e(z, y)$, replacing $z$ (which is not a head variable) by $\_$.  We thus obtain the rule 
\[ R_2' : \mathsf{TC}_{\mathsf{TC}(x, y) \gets e(x, \_), e(\_, y)}(x, y) \gets \mathsf{TC}_{\mathsf{TC}(x, y) \gets e(x, y)}(x, z), e(z, y), \]
which we add to~$\Pi'$.

We can now  add the adornment 
$\mathsf{TC}(x, y) \gets e(x, \_), e(\_, y)$ to the IDB predicate in the 
rule body of $R_2$. That is, we assume that, when executing program $\Pi'$ on a database $D$,
the body atom will be generated by firing rule $R'_2$. Hence, in the first place, we have to 
rename the variables in~$R'_2$ apart from the variables in $R_2$. However, we then need to unify the head of $R'_2$ with the body atom in $R_2$. 
Then, inlining results in an adornment of the form $ \mathsf{TC}(x, y) \gets e(x, \_), e(\_, u), e(u, y) $, 
which we transform into $ \mathsf{TC}(x, y) \gets e(x, \_), e(\_, \_), e(\_, y)$ and further into $ \mathsf{TC}(x, y) \gets e(x, \_), e(\_, y)$ by applying $\galgo$. 
So we add to $\Pi'$ the rule 
   \[ R'_3 : \mathsf{TC}_{\mathsf{TC}(x, y) \gets e(x, \_),  e(\_, y)}(x, y) \gets \mathsf{TC}_{\mathsf{TC}(x, y) \gets e(x, \_), e(\_, y)}(x, z), e(z, y). \]

In principle, we could now try to 
use the newly created rule $R'_3$ to generate yet another adorned version of $R_2$ by unifying the head of 
$R'_3$ with the IDB body atom of $R_2$. It is easy to verify that the same process as described above would again yield 
an adorned rule that is exactly the same as $R_3'$.
Hence, the process of generating new EDB-adorned rules has reached a fixpoint, and we can take $\Pi' = \{R'_1,R'_2,R'_3\}$, which was already shown in Example~\ref{ex:tc-intro},
as our final EDB-bounded program.
\end{example}

As was illustrated in Example~\ref{ex:propagation-falgo}, the use of an adornment relaxation function allows us to delete redundant body atoms from an 
adornment. However, in the context of an adorned program, an entire rule may be redundant in that its deletion does not alter the 
semantics of the adorned program. To detect certain forms of redundancy of an adorned rule for potentially different use cases, we 
introduce ``membership checking functions'' next:

\begin{definition}[membership checking function]
\label{def:membership-checking}
    A \emph{membership checking function} is a function $h$ which takes as input an EDB-bounded rule $R$ with head predicate $q_\rho$, 
    and a datalog program $\Pi$, and outputs either true or false. In this work, we consider two membership checking functions:
    \begin{itemize}[leftmargin=*]
        \item $h_{\sf eq}$: A function that returns true if there exists some $R' \in \Pi$ that is identical to $R$, up to variable renaming, and false otherwise. Except for the discussion in \cref{sec:boundedness}, this function is used throughout the remainder of the paper for deriving size bounds.
        \item $h_{\sf cont}$: A function that returns true if
        either (1) $q_{\rho}$ is not recursive in $\Pi$ and there exists a predicate $q_{\rho'}$ such that $\rho'$ subsumes $\rho$,
        or (2) there exists some $R' \in \Pi$ that is identical to $R$ up to variable renaming.
        Otherwise, it returns false. This function will be used in \cref{sec:boundedness} in the semi-decision procedure for boundedness.
    \end{itemize}
\end{definition}

\subsubsection{Algorithm to Construct EDB-Bounded Programs}

We now present \cref{alg:adornment}, which constructs an EDB-bounded datalog program based on a given program using a given adornment relaxation function $g$ and membership checking function $h$.
The idea of the algorithm aligns with the intuition given in our previous example.
Namely, it iteratively tries, until a fixpoint is reached, to turn the rules in the original program $\Pi$ into adorned rules in $\Pi'$ by replacing each IDB body atom of the original rule by an adorned atom already in $\Pi'$, and then generating the adornment for the head atom using the adornment relaxation function $g$. The adorned rule is then added to $\Pi'$ subject to the membership checking function $h$.

\begin{algorithm}
\caption{Construction of an EDB-bounded datalog program from a given program}\label{alg:adornment}
\SetKwInOut{Input}{input}
\SetKwInOut{Parameters}{parameters}
\SetKwInOut{Output}{output}
\SetKwInOut{Known}{known}
\Input{A datalog program $\Pi$}
\Parameters{An adornment relaxation function $g$, a membership checking function $h$}
\Output{An EDB-bounded datalog program $\Pi'$}
Initialize $\Pi'$ as an empty datalog program \;
\Repeat{fixpoint for $\Pi'$}{
    \ForEach{rule $R$ in $\Pi$}{
        W.l.o.g. assume that $R$ has the form $q(\bar{s}) \gets (\bigwedge_{i = 1}^{\lambda} p_i(\bar s) ) \wedge \mathcal{E}$, where $p_1(\bar s_1), \dots, p_\lambda(\bar s_\lambda)$ are IDB atoms, and $\mathcal{E}$ is a conjunction of EDB atoms \;
        \ForEach{$(\rho_1, \dots, \rho_{\lambda})$, such that for all $i = 1, \dots, \lambda$, there exists a rule $R_i$ in $\Pi'$ with head atom ${p_i}_{\rho_i}(\bar{v}_i)$ with adornment $\rho_i$}{
            Apply a variable renaming to all head atoms ${p_i}_{\rho_i}(\bar{v}_i)$ to obtain ${p_i}_{\rho_i'}(\bar{v}_i')$ so that 
            $R$ and all head atoms in all $R_i$ are variable-disjoint \;
            Let $\rho_i'$ be of the form $p_i(\bar v_i') \gets \mathcal{B}_i$ \;
            $\sigma \gets \textrm{some mgu} \{ ( \bar{s}_i, \bar{v}_i' ) \mid 1 \leq i \leq \lambda \}$ \;
            \If{no $\sigma$ exists}{
                Continue \;
            }
            $\rho_0 :=$ the rule $q(\sigma(\bar s)) \gets (\bigwedge_{i = 1}^{\lambda} \sigma(\mathcal{B}_i)) \wedge \sigma(\mathcal{E})$ \;
            $\rho := g(\rho_0)$ \;
            $R' := \text{the rule } q_{\rho}(\sigma(\bar{s})) \gets (\bigwedge_{i = 1}^{\lambda} {p_i}_{\rho_i}(\sigma(\bar{s}_i))) \wedge \sigma(\mathcal{E})$ \;
            \If{$h(R', \Pi')$ is false}{
                Add $R'$ to $\Pi'$ \;
            }
        }
    }
}

\Return {$\Pi'$}

\end{algorithm}

\label{sec:app-alg-description}

We give a detailed explanation of the algorithm using the adornment relaxation function $\galgo$ and membership checking function $h_{\sf eq}$. We examine the behavior of the algorithm by distinguishing rules (in $\Pi$) with only EDB body atoms and rules containing at least one IDB atom in the body.

\paragraph{Base case}
Let $R\colon q(\bar s) \leftarrow p_1(\bar s_1), \dots, p_\mu(\bar s_\mu)$
be a rule in $\Pi$ with only EDB atoms in the body. 
Then we construct a rule
$R'\colon q_\rho(\bar s) \leftarrow p_1(\bar s_1), \dots, p_\mu(\bar s_\mu)$
to be added to $\Pi'$.
To obtain the adornment $\rho$, let $\rho_0$ be a rule of the form
$\rho_0 \colon q(\bar s) \leftarrow p_1(\bar t_1), \dots, p_\mu(\bar t_\mu)$
with $t_{ij} = s_{ij}$ if $s_{ij}$ is a head variable in $R$ (i.e., 
$ t_{ij} \in \var(\bar s)$) 
and $t_{ij} = \_$  otherwise.
We then obtain $\rho$ from $\rho_0$ by applying the adornment relaxation function $\galgo$. 

\paragraph{Adornment propagation}
Let $R\colon q(\bar s) \leftarrow p_1(\bar s_1), \dots, p_\mu(\bar s_\mu)$
be a rule in $\Pi$, where at least one $p_i$ is an IDB predicate. 
Then, for every $i$, such that $p_i$ is an IDB predicate, choose one
    rule $R_i$ in $\Pi'$
    with head atom ${p_i}_{\rho_i}(\bar v_i)$ for some adornment $\rho_i$.
If, for some $i$, no such rule exists yet, then no new EDB-adorned datalog rule 
    can be created from rule $R$ at this point.  
Otherwise, let us assume w.l.o.g. that $R$ has the IDB body atoms 
$p_1(\bar s_1), \dots, p_\lambda(\bar s_\lambda)$ for some $\lambda \leq \mu$ 
and that the remaining body atoms are EDB atoms.

Now apply to all head atoms ${p_i}_{\rho_i}(\bar v_i)$  of the 
rules $R_1, \dots, R_\lambda$ a variable renaming 
to make sure that rules $R$ and the head  atoms 
of the rules $R_1, \dots, R_\lambda$ 
are variable-disjoint. Note that, for each atom ${p_i}_{\rho_i}(\bar v_i)$, 
we now apply the same  renaming of the variables $\bar v_i$ 
also  to  the adornment $\rho_i$. Hence, 
by this variable 
renaming,  the new head 
atoms of the rules $R_1, \dots, R_\lambda$
are now  
${p_i}_{\rho'_i}(\bar v'_i)$.

Next, we compute the 
simultaneous, most general unifier $\sigma$ 
of $\{ (\bar s_1, \bar v'_1), \dots, (\bar s_\lambda, \bar v'_\lambda) \}$.
If no such simultaneous unifier exists, then 
no new EDB-adorned datalog rule 
can be created from rule $R$ with this choice of EDB-adorned rules 
$R_1, \dots, R_\lambda$.
Otherwise, we add to $\Pi'$ an adorned rule 
$$R'\colon q_\rho(\sigma(\bar s)) \leftarrow 
{p_1}_{\rho_1} (\sigma(\bar s_1)), \dots, {p_\lambda}_{\rho_\lambda} (\sigma(\bar s_\lambda)),
p_{\lambda+1}( \sigma(\bar s_{\lambda+1})), \dots, p_\mu(\sigma(\bar s_\mu)),
$$
where 
we first define the rule $\rho_0$ and then obtain 
$\rho$ by applying on $\rho_0$ 
the adornment relaxation function $\galgo$.
The rule $\rho_0$ is constructed as follows: 

\begin{itemize}[leftmargin=*]
    \item The head of $\rho_0$ is $q(\sigma(\bar s))$, i.e., the same as the 
    unadorned head of the new
    rule $R'$;
    \item for every EDB atom $p_i(\bar s_i)$ in the body of $R$, we add the atom 
    $p_i(\bar t_i)$ to the body of $\rho_0$
with $t_{ij} = \sigma(s_{ij})$ if 
$t_{ij} \in \var(\sigma(\bar s))$
(i.e., $\sigma(s_{ij})$ is a head variable of $R'$),  and 
$t_{ij} = \_$  otherwise;
    \item for every  $i \in \{1, \dots, \lambda\}$ and
    for every atom $e(\bar a)$ in the body of $\rho'_i$, 
    we add to the body of $\rho_0$ the atom $e(\bar b)$ with $b_j = \sigma (a_j)$ 
    if $\sigma(a_j) \in \var(\sigma(\bar s))$
     (i.e., $\sigma(a_j)$ is a  head variable of $R'$), 
     and  $b_j = \_$ otherwise.
\end{itemize}

% \smallskip
% \noindent{\em membership checking function.}
\paragraph{Membership checking}
Suppose that a new adorned rule is generated as explained above. Then our algorithm only adds this new rule to $\Pi'$ if 
the membership checking function $h_{\sf eq}$ returns false---i.e., the new rule is not just a renaming of an already existing rule.

\medskip

\noindent
Let us comment on two particular aspects of the construction of $\Pi'$ 
in  Algorithm~\ref{alg:adornment}: 
\begin{enumerate}[leftmargin=*]
    \item 
It should be noted that, as adornment of an IDB body atom ${p_i}_{\rho_i} (\sigma(\bar s_i))$, 
we take the original adornment $\rho_i$ of the head atom of rule $R_i$. In contrast, when
transferring body atoms from the adornment of the head atom of rule $R_i$ to the body
of the adornment $\rho_0$, we first take the adornment $\rho'_i$ resulting from the variable
renaming and then also apply the mgu $\sigma$  to these body atoms. 
We then apply $\galgo$ to make sure that 
all variables in an EDB atom in the adornment $\rho_0$ also occur in the head of rule $R'$ 
(after applying the 
mgu $\sigma$ to the rule $R$) and that redundant EDB atoms are removed from the rule body of $\rho_0$. 
\item 
Note that we have explicitly mentioned the base case here for the sake of readability. 
Clearly, the case of a rule with only EDB atoms in the 
body is a special case of the adornment propagation where $\lambda = 0$. 
This is how rules with only EDB atoms in the body are handled in Algorithm~\ref{alg:adornment}.
\end{enumerate}

\subsubsection{Requirement of Simultaneous Unifiers in \cref{alg:adornment}}
\label{sec:app-sim-unifier}

Note that, in \cref{alg:adornment}, we have to compute
a simultaneous unifier between each IDB body atom of a rule $R$ and the rule heads of the adorned rules
$R_1, \dots, R_\lambda$.
We show why this is required by applying the construction of Algorithm~\ref{alg:adornment}
to a slightly more complex datalog program as an example.

\begin{example}
Suppose that we apply Algorithm~\ref{alg:adornment} with functions $\galgo$ and $h_{\sf eq}$. 
We consider the following datalog program 
        \begin{alignat*}{1}
        p(v, v) & \gets e(v, w).\\
        q(x, y, z) & \gets p(x, y), p(y,z).
        \end{alignat*}
Note that, in this case, we in fact have $x = y$ and $y = z$. Therefore, we should ultimately get the following adorned program:
        \begin{align*}
        p_{p(v,v) \gets e(v,\_)}(v, v) &\gets E(v, w).\\
        q_{q(x,x,x) \gets e(x,\_)}(x, x,x) &\gets p_{p(v,v) \gets e(v,\_)}(x, x),
        p_{p(v,v) \gets e(v,\_)}(x, x).      
        \end{align*}
Since the two body atoms in the second rule are identical, we can actually delete one 
and get 
        \begin{align*}
        q_{q(x,x,x) \gets e(x,\_)}(x, x,x) &\gets p_{p(v,v) \gets e(v,\_)}(x, x).
        \end{align*}

Now, how should we define the adornment propagation in this case---taking variable renaming and unification into account, and being able to reason about the equality of variables across different atoms (e.g., $x$, $y$, and $z$ in this particular example)?

First, we choose a rule for each of the two IDB body atoms. In this case, $R_1 = R_2$.
So the adornments of the two IDB body atoms are fixed, namely $p(v,v) \gets e(v,\_)$ for both.
However, before we can compute the simultaneous mgu of each body atom with the
    head atom of the corresponding rule, we have to rename the variables of all rules apart to avoid collisions, i.e., we get 
        \begin{align*}
            R'_1 &: p_{p(v,v) \gets e(v,\_)}(v, v) \gets e(v, w). \\
            R'_2 &: p_{p(v',v') \gets e(v',\_)}(v', v') \gets e(v', w'). 
        \end{align*}
    Note that, for computing the adornment of the head atom of the new rule, we have also applied the variable renaming in rule $R_2$ 
    to the adornment. This is in contrast to the adornment of the second body atom of the new rule,
    where we have left the adornment of $R_2$ unchanged.

Now we have to compute $\sigma = {\rm mgu} ( \{(p(x,y), p(v,v)), (p(y,z), p(v',v'))\})$.
    Clearly, $\sigma$ sets all variables $x,y,z,v,v'$ equal. Nevertheless, there are several
    ways of expressing $\sigma$, e.g. 
        \begin{align*}
            \sigma_1 &= \{x \mapsto v, y \mapsto v, z \mapsto v, v' \mapsto v  \}. \\
            \sigma_2 &= \{y \mapsto x, z \mapsto x, v \mapsto x, v' \mapsto x  \}. 
        \end{align*}

We now have to apply the mgu to all atoms (head atom and all body atoms) of the rule $R$, i.e., depending on the choice of mgu, we get: 
        \begin{alignat*}{2}
        \mbox{With mgu $\sigma_1$: } \quad &&
        q_\rho(v, v, v) &\ \gets p_{\rho_1}(v, v), p_{\rho_2}(v,v).  \\
        \mbox{With mgu $\sigma_2$: }  \quad &&
        q_\rho(x, x, x) &\ \gets p_{\rho_1}(x, x), p_{\rho_2}(x,x).  
        \end{alignat*}
        
        Since $\rho_1 = \rho_2$, we may  drop one of the two body atoms
        in each of the above two rules.

Finally, we have to compute the adornment of the head atom of the new rule. 
We thus construct a rule whose head atom is the same as the head atom of the new rule; 
as body
atoms, we collect the body atoms of the adornments of the head atoms in the rules $R_1$ and $R_2$, but first apply the mgu to these body atoms. We thus get the following adornments: 
        \begin{alignat*}{2}
        \mbox{With mgu $\sigma_1$: }  \quad
        & \rho =\ & q(v, v, v) & \gets e(v, \_),  e(v, \_).  \\
        \mbox{With mgu $\sigma_2$: }  \quad
        & \rho =\ & q(x, x, x) & \gets  e(x, \_),  e(x, \_).
        \end{alignat*}
In both rules, the body consists of two identical atoms. Hence, 
in each rule, we can delete one of the two body atoms, and we arrive at the 
rule      $q_{q(x,x,x) \gets e(x,\_)}(x, x,x) \gets$ $p_{p(v,v) \gets e(v,\_)}(x, x)$,
as mentioned above.

\end{example}

\subsubsection{Equivalence of EDB-Bounded Programs Produced by \cref{alg:adornment}}

We now establish the close relationship between the original datalog program $\Pi$ and the 
adorned program $\Pi'$ resulting from Algorithm~\ref{alg:adornment}. More specifically,
we show that, by an appropriate choice of the functions $g$ and $h$, the two programs 
are ``equivalent'', i.e.,   
for any EDB, they compute ``the same'' IDB relations modulo adornment. 
This notion of ``equivalence'' is formally defined below.

\begin{definition}[equivalent EDB-bounded program]
\label{def:equivalent}
    Let $\Pi$ be a datalog program and let $\Pi'$ be an EDB-bounded datalog program. 
    We say that 
    {\em $\Pi$ and $\Pi'$ are equivalent}, if for every EDB $D$ and for every 
    IDB predicate $q$ occurring in $\Pi$, the following condition holds: 
    $$
    \sem{q(\Pi,D)}   =
    \bigcup_{\rho}\,  \sem{q_\rho(\Pi',D)}
    $$
    That is, the set of tuples in the IDB relation $q$ that can be derived by program $\Pi$ from the database $D$ coincides with the union over all sets of tuples in any adorned 
    IDB relation $q_\rho$ that can be derived by program $\Pi'$ from
    the database $D$. 
\end{definition}

We next show that the EDB-adorned program $\Pi'$ constructed from an arbitrary datalog program
$\Pi$ according to \cref{alg:adornment}, with $\galgo$ and $h_{\sf eq}$, is indeed  an {\em equivalent EDB-bounded} program.

\begin{theorem}
\label{thm:equivalent}
Given an arbitrary datalog program $\Pi$, let $\Pi'$ be the EDB-adorned program
obtained from $\Pi$ according to \cref{alg:adornment} with adornment relaxation function $\galgo$ and 
membership checking function $h_{\sf eq}$.
Then $\Pi'$ is an EDB-bounded program equivalent to $\Pi$.
\end{theorem}

\begin{proof}
Let $\Pi$ and $\Pi'$ be according to the statement of the theorem. 
Moreover, let $D$ be an arbitrary database over the EDB predicates of $\Pi$ (and, hence, 
also of $\Pi'$). We have to prove
(1) that the two programs lead to the same IDB relations (modulo adornments) 
and (2) that every adorned rule in $\Pi'$ is indeed bounded by the adornment of the head predicate.

\paragraph{Equivalence}
Let $q$ be an arbitrary IDB predicate occurring in $\Pi$. 
We have to show that $\sem{q(\Pi,D)}   =
\bigcup_{\rho}\,  \sem{q_\rho(\Pi',D)}$ holds. 

The inclusion $\sem{q(\Pi,D)}   \supseteq
\bigcup_{\rho}\,  \sem{q_\rho(\Pi',D)}$  is seen as follows: 
Suppose that $\Pi'$ contains a rule 
$R'\colon q_\rho(\sigma(\bar s)) \leftarrow {p_1}_{\rho_1}(\sigma(\bar s_1)), \dots, {p_\mu}_{\rho_\mu}(\sigma(\bar s_\mu))$, 
then,
by the construction in Algorithm~\ref{alg:adornment},
$\Pi$ contains the rule 
$R\colon q(\bar s) \leftarrow p_1(\bar s_1), \dots, p_\mu(\bar s_\mu)$.
Hence, any proof tree $\mathcal{T}'$ of a ground IDB atom $q_\rho(\bar t)$ with program 
$\Pi'$ can 
be turned into a proof tree $\mathcal{T}$ of $q(\bar t)$ by simply 
replacing every rule $R'$ by the corresponding rule $R$ and removing all adornments.

It remains to prove the inclusion $\sem{q(\Pi,D)}   \subseteq
\bigcup_{\rho}\,  \sem{q_\rho(\Pi',D)}$. That is, 
for every $\bar t \in \sem{q(\Pi,D)}$, we have to show that 
there exists an adorned predicate $q_\rho$ with  
$\bar t \in \sem{q_\rho(\Pi',D)}$. We prove this property by induction 
on the depth of a proof tree of $q(\bar t)$ of program $\Pi$ over the  database~$D$. 

\smallskip
\noindent
{\em Base case.}
Suppose that the depth of a proof tree of $q(\bar t)$ is $j = 1$, i.e., 
$q(\bar t)$ is derived by applying a rule of the form 
$q(\bar s) \leftarrow p_1(\bar s_1), \dots, p_\mu(\bar s_\mu)$,
such that all body atoms are EDB atoms. 
Then $\Pi'$ contains a rule 
$q_\rho(\bar s) \leftarrow p_1(\bar s_1), \dots, p_\mu(\bar s_\mu)$,
where $\rho$ is constructed according to the base case of 
Algorithm~\ref{alg:adornment}.
In particular, $\rho$ is 
such an adornment with $\bar t \in \sem{q_\rho(\Pi',D)}$.

\smallskip
\noindent
{\em Induction step.}
Now suppose that the depth of a proof tree $\mathcal{T}$ of $q(\bar t)$ is $j > 1$.
That is, the proof tree $\mathcal{T}$ has 
$q (\bar t)$ as root node, which is connected via edges with some label $R$ to its child nodes. Let 
$p_1(\bar t_1), \dots, p_\mu(\bar t_\mu)$ be the ground atoms labeling the child nodes of 
the root in $\mathcal{T}$ and let the rule 
$R$ have the form 
$q(\bar s)  \leftarrow p_1(\bar s_1), \dots, p_\mu(\bar s_\mu)$.
Hence, 
$R$ fires by (simultaneously) instantiating each body atom $p_i(\bar s_i)$ to 
$p_i(\bar t_i)$, i.e., there exists a substitution $\eta$, such that 
$\eta(\bar s_i) = \bar t_i$ for every $i \in \{1, \dots, \mu\}$
and $\eta(\bar s) = \bar t$.

Clearly, the ground atoms $p_1(\bar t_1), \dots, p_\mu(\bar t_\mu)$
are either EDB atoms or they have a proof tree of depth $< j$.
Since we are assuming that the depth of $\mathcal{T}$ is $> 1$, there is at least 
one IDB atom in the body of $R$.
W.l.o.g., let $1 \leq \lambda \leq \mu$ such that $p_1, \dots, p_\lambda$ are 
IDB predicates and $p_{\lambda+1}, \dots, p_\mu$ are EDB predicates.
Hence, by the 
induction hypothesis, we can derive adorned atoms
${p_i}_{\rho_i}(\bar t_i)$ for some adornments $\rho_i$ also in $\Pi'$ with 
$1 \leq i \leq \lambda$.
In particular, this means that $\Pi'$ contains rules 
$R'_i, \dots, R'_\lambda$ with rule heads ${p_i}_{\rho_i}(\bar u_i)$, such that 
${p_i}_{\rho_i}(\bar t_i)$ is a ground instance of 
${p_i}_{\rho_i}(\bar u_i)$. 
W.l.o.g., suppose that $R'_i, \dots, R'_\lambda, R$ are variable disjoint (otherwise apply variable
renamings to $R'_i, \dots, R'_\lambda$ as in Algorithm~\ref{alg:adornment}).
Clearly, $\bar t_i$
is a common ground instance
of $\bar u_i$ and $\bar s_i$, for every $i \in \{1, \dots, \lambda\}$.
Hence, we can simultaneously unify 
$\bar u_i$ and $\bar s_i$
for all $1 \leq i \leq \lambda$.
Let $\sigma = {\rm mgu} \{ (\bar u_1, \bar s_1), \dots, (\bar u_\lambda, \bar s_\lambda) \}$.
Then there exists a substitution $\tau$, such that  
$\tau\sigma(\bar u_i) = \tau\sigma(\bar s_i) = \bar t_i = \eta (\bar s_i)$ for every $i \in \{1, \dots, \lambda\}$.
Note that $\eta$ possibly instantiates additional variables compared to $\eta(\sigma(\cdot))$, since
it also maps all tuples $\bar s_i$ with $i \in \{\lambda+1, \dots, \mu\}$ to $\bar t_i$.
Hence, there exists a substitution $\theta$ (equal to $\eta$  or extending $\eta$),
such that $\theta(\sigma(\bar s_i)) = \eta( \bar s_i) = \bar t_i$ for every $i \in \{1, \dots, \mu\}$
and $\theta(\sigma(\bar s)) = \eta( \bar s) = \bar t$.

By Algorithm~\ref{alg:adornment},
$\Pi'$ contains the rule $R'\colon q_\rho(\sigma(\bar s)) \leftarrow 
{p_1}_{\rho_1} (\sigma(\bar s_1)), \dots, {p_\mu}_{\rho_\mu} (\sigma(\bar s_\mu))$, 
where $\rho$ is constructed according to the adornment propagation 
in Algorithm~\ref{alg:adornment}. By the above considerations, 
$\theta$ maps each body atom ${p_i}_{\rho_i}(\sigma(\bar s_i))$ to
 ${p_i}_{\rho_i}(\theta(\sigma(\bar s_i))) =  {p_i}_{\rho_i}(\bar t_i)$ and it maps 
 the 
 head atom  $q_{\rho}(\sigma(\bar s))$ to $q_{\rho}(\theta(\sigma(\bar s))) = q_{\rho}(\bar t)$.
Since we are assuming that  all the EDB-adorned atoms ${p_i}_{\rho_i}(\bar t_i)$  for $i \in \{1, \dots, \lambda\}$
can be derived 
by $\Pi'$ from the EDB $D$,
and $D$ contains the ground atoms $p_i(\bar t_i)$  for $i \in \{\lambda + 1, \dots, \mu\}$,
we conclude that 
$q_\rho(\bar t)$ can indeed be derived by $\Pi'$ from the EDB $D$.

\paragraph{Boundedness}
We have to show that every rule in $\Pi'$ is 
bounded by the adornment of its head predicate. That is,  
let $R'$ with head predicate $q_\rho$ be a rule in $\Pi'$ and let $q_\rho(\bar t)$ be derived by rule $R'$ when 
evaluating $\Pi'$ over an EDB $D$. We have to show that $q(\bar t)$ is also derived when 
evaluating the single-rule program consisting of rule $\rho$ over the EDB $D$.
The proof is by induction on the depth of a proof tree in $\Pi'$ of a fact $q_\rho(\bar t)$.

\smallskip
\noindent
{\em Base case}. Suppose that $q_\rho (\bar t)$ is derived by a proof tree $\mathcal{T}$ of depth $j = 1$ 
and that 
rule $R'$ with head predicate $q_\rho$ is the top-most rule in $\mathcal{T}$. 
Then $R'$ is  
of the form $q_\rho(\bar s) \gets p_1(\bar s_1), \dots, p_1(\bar s_\mu)$, 
where all $p_i$ are EDB predicates.
By the base case 
in Algorithm~\ref{alg:adornment}, $R'$ is constructed from a rule  
$R$ of the form $q(\bar s) \gets p_1(\bar s_1), \dots, p_1(\bar s_\mu)$ in $\Pi$, 
such that $\rho$ is obtained from $R$ by leaving the head atom unchanged and replacing in every body atom all arguments that are not head variables by ``$\_$'' and then 
deleting redundant body atoms. Recall that the occurrences of ``$\_$'' in body atoms are treated like pairwise distinct, existentially quantified variables. Hence, whenever $R$ fires, so does $\rho$, since the body of $\rho$ has been 
constructed by starting off with the body of $R$ and possibly ``weakening'' (i.e., replacing some arguments by ``$\_$'')
or even deleting some of the body atoms.

\smallskip
\noindent
{\em Induction step}. 
Suppose that $q_\rho (\bar t)$ is derived by a proof tree $\mathcal{T}$ of depth $j > 1$ 
and that 
rule $R'$ with head predicate $q_\rho$ is the top-most rule in $\mathcal{T}$. 
That is, $\mathcal{T}$ has $q_\rho(\bar t)$ as root node, which is connected via edges with label $R'$ 
to its child nodes. Moreover, since $j > 1$, $R'$ contains at least one IDB atom in the body. 
By the the adornment propagation in Algorithm~\ref{alg:adornment},
rule $R'$ is of the form 
$$R'\colon q_\rho(\sigma(\bar s)) \leftarrow 
{p_1}_{\rho_1} (\sigma(\bar s_1)), \dots, {p_\lambda}_{\rho_\lambda} (\sigma(\bar s_\lambda)),
p_{\lambda+1}( \sigma(\bar s_{\lambda+1})), \dots, p_\mu(\sigma(\bar s_\mu)).
$$
Hence, the child nodes of the root in  $\mathcal{T}$ are labeled by ground atoms  
${p_1}_{\rho_1} (\bar t_1)$, \dots, ${p_\lambda}_{\rho_\lambda} (\bar t_\lambda),
p_{\lambda+1}( \bar t_{\lambda+1})$, \dots, $p_\mu(\bar t_\mu)$, such that 
 ${p_1}_{\rho_1}, \dots {p_\lambda}_{\rho_\lambda}$ are IDB predicates, and 
$p_{\lambda +1}, \dots, p_\mu$ are EDB predicates. 
Importantly, all tuples $\sigma(\bar s_i)$ with $i \in \{1, \dots, \mu\}$ can be simultaneously instantiated to $\bar t_i$. 
Hence, there exists a substitution $\tau$, such that 
$\bar t_i = \tau(\sigma(\bar s_i))$ for every $i \in \{1, \dots, \mu\}$
and $\bar t = \tau(\sigma(\bar s))$.

Moreover, there exist rules $R_1, \dots, R_\lambda$ in $\Pi'$ 
with head predicates ${p_1}_{\rho_1}, \dots {p_\lambda}_{\rho_\lambda}$,
which were used in the construction of rule $R'$.
Recall from the propagation of adornments in Algorithm~\ref{alg:adornment} that
$\sigma$ is actually obtained as mgu when simultaneously unifying, for every $i \in \{1, \dots, \lambda\}$, 
the arguments $\bar u_i$ of the head atom ${p_i}_{\rho_i}(\bar u_i)$ of rule $R_i$ with the 
arguments $\bar s_i$ of the corresponding body atom $p_i(\bar s_i)$ of the rule $R$ in $\Pi$.
(Here we are assuming w.l.o.g., that the rules $R_1,\dots, R_\lambda,R$ are pairwise variable-disjoint; 
this can always be achieved by appropriate renaming of variables.)
Hence, $\sigma(\bar s_i) = \sigma(\bar u_i)$ and, therefore, 
we also have $\bar t_i = \tau(\sigma(\bar u_i))$ for every $i \in \{1, \dots, \lambda\}$.

In the proof tree $\mathcal{T}$, each ${p_i}_{\rho_i} (\bar t_i)$ has itself a proof tree $\mathcal{T}_i$
of depth $< j$,
and the root node of $\mathcal{T}_i$ is connected with edges labeled by $R_i$ with its child nodes. 
Hence, by the 
induction hypothesis, every $p_i (\bar t_i)$ is also derived when evaluating the 
single-rule program consisting of rule $\rho_i$ only. By the construction of $\Pi'$ in 
Algorithm~\ref{alg:adornment},
rule $R_i$ and the rule $\rho_i$ adorning the head predicate of $R_i$ have, apart from the adornment,
the same head atom. That is, the head atom of $\rho_i$ is $p_i(\bar u_i)$. Moreover, as has been 
argued above, we have $p_i(\bar t_i) = p_i(\tau(\sigma( \bar u_i)))$. That is, the rule $\rho_i$ fires on 
the EDB $D$ by applying a substitution of the form $\eta_i \cup \tau(\sigma(\cdot))$ to $\rho_i$, where 
$\eta_i$ extends the substitution $\tau(\sigma(\cdot))$ to the variables only occurring in the body of $\rho_i$. 
Of course, for any two indices $i \neq i'$, the substitutions 
$\eta_i$ and $\eta_{i'}$ cannot interfere with each other, since we are assuming that any two 
rules out of $R_1, \dots, R_\lambda$ are 
variable-disjoint and, therefore, also the adornments $\rho_1, \dots, \rho_\lambda$
(which may only contain variables occurring  in the head of the corresponding rules
$R_1, \dots, R_\lambda$) 
are pairwise variable-disjoint. 

Recall from the adornment propagation in Algorithm~\ref{alg:adornment} that the adornment $\rho$
of the head atom $q_\rho(\sigma(\bar s))$ of rule $R'$ is obtained as follows: 
\begin{itemize}[leftmargin=*]
    \item The head atom of $\rho$ is the same as the head atom of $R'$ without the adornment, i.e., 
    $q(\sigma(\bar s))$;
    \item in every EDB atom of the body of $R'$, the variables not occurring in $\sigma(\bar s)$ are replaced by 
    ``$\_$'', and the resulting atom is added to the body of $\rho$;
    \item in every body atom of every rule $\rho_i$, we first apply the substitution $\sigma$ and then again 
    replace the variables not occurring in $\sigma(\bar s)$ by 
    ``$\_$'', before we add the resulting atom to the body of $\rho$.  
\end{itemize}

Actually, some of the aforementioned body atoms of rule $\rho$ may possibly be deleted by the redundancy elimination step. 
Anyway, we show that all of the above mentioned (potential) body atoms of rule $\rho$ can be simultaneously matched into the 
EDB $D$ by a substitution that extends $\tau(\sigma(\cdot))$. 

    First consider some body atom of $\rho$ obtained from an EDB atom $p_i(\sigma(\bar s_i))$ of the body of $R'$. Even without replacing some of the arguments in 
    $\sigma(\bar s_i)$  by ``$\_$'' (i.e., by a fresh existentially quantified variable), we have that $p_i(\tau(\sigma(\bar s_i)))$ is a ground atom from the EDB $D$. Hence,
    there clearly exists an extension of $\tau(\sigma(\cdot))$ to the fresh, existentially quantified variables so that the atom obtained from $p_i(\sigma(\bar s_i))$
    by introducing ``$\_$'' is mapped to a ground atom in the EDB.
    
    Now consider a body atom of $\rho$ that is obtained from some body atom of some rule $\rho_i$. As has been shown above, even without replacing 
    some of the arguments in this atom by ``$\_$'' (i.e., by a fresh existentially quantified variable), $\tau(\sigma(\cdot))$ can be extended by a 
    substitution $\eta_i$ that maps every body atom of $\rho_i$ to a ground atom from the EDB $D$. Hence, we may again conclude that
    there exists an extension of $\tau(\sigma(\cdot))$ to the fresh, existentially quantified variables so that every body atom of $\rho_i$ is mapped to 
    an atom from the EDB $D$.

In other words, the rule 
$\rho$ fires on the EDB $D$ and produces the atom 
 $q(\bar t) = q(\tau(\sigma(\bar s)))$.
\end{proof}

\section{Semi-Deciding and Rewriting Bounded Datalog Programs}
\label{sec:boundedness}

In this section, we revisit the problem of datalog boundedness.
We first show the semi-decidability of boundedness by (potentially infinite) inlining.
Then, as one use case of our framework, we give an instantiation that rewrites most practical bounded datalog programs into equivalent UCQs by simple inlining.
In case the equivalent UCQ is too large, or the program is unbounded, the algorithm applies adornment relaxation functions to obtain an equivalent finite EDB-bounded program.

We start by formalizing this idea of potentially infinite inlining of datalog programs:
\begin{definition}[unrolling]
    The \emph{unrolling} of a datalog program $\Pi$ is defined as the non-recursive, potentially infinite equivalent program derived by repeatedly inlining the rules of $\Pi$.
\end{definition}

It has also been well known that
\begin{theorem}[{\cite[Theorem~3]{10.1145/322217.322221}}]
\label{thm:ucq_containment}
    Given two UCQs $Q_1$ and $Q_2$, then $Q_1 \subseteq Q_2$ if and only if every CQ in $Q_1$ is contained in a CQ in $Q_2$.
\end{theorem}
Notice that \autoref{thm:ucq_containment} can be easily generalized to infinite UCQs and to non-recursive datalog programs, which are equivalent to finite sets of UCQs. This leads directly to the following corollary:
\begin{corollary}
\label{prop:equivalent_non_rec}
    Given a datalog program $\Pi$, let $\Pi^*$ be the unrolling of $\Pi$. Let $\Pi_b$ be a non-recursive (finite) datalog program equivalent to $\Pi$. Then $\Pi_b$ is equivalent to a finite subset of rules from $\Pi^*$.
\end{corollary}
\begin{proof}
    A set of datalog rules sharing the same head predicate defines a UCQ. Consequently, every CQ derived from rules in $\Pi_b$ must be contained within a CQ from the rules in $\Pi^*$. Since $\Pi_b$ has a finite set of rules, the set of CQs equivalent to $\Pi_b$ in $\Pi^*$ is also finite.
\end{proof}
Next, we show an application of this result within our framework to obtain a non-recursive program that is 
equivalent to a given bounded datalog program.
\begin{proposition}
\label{prop:equivalent_non_rec_algo}
    Given a datalog program $\Pi$, let $\Pi'$ be the (potentially infinite) output of applying \cref{alg:adornment} on $\Pi$ with the identity function $\it id$ as the adornment relaxation function, and let $\Pi^*$ be the non-recursive unrolling of $\Pi$. Then $\Pi'$ is equivalent to $\Pi^*$.
\end{proposition}
\begin{proof}
    The head adornments form the unrolling of $\Pi$ (each of which is simply an inlining), and every rule in $\Pi^*$ is equivalent to its head adornment.
\end{proof}

Furthermore, by applying $h_{\sf cont}$ as the membership checking function, which attempts not to add an adornment if it is subsumed by another adornment, then the algorithm always terminates when the input program is bounded:
\begin{proposition}
\label{prop:bounded_eq_terminate}
Algorithm \ref{alg:adornment}, with $\it id$ as the adornment relaxation function and $h_{\sf cont}$ as the membership checking function, 
terminates on a datalog program $\Pi$ if and only if $\Pi$ is bounded.
\end{proposition}
\begin{proof}
    Let $\Pi_b$ be a non-recursive datalog program equivalent to $\Pi$ and let $\Pi'$ be the output of the new algorithm. 
By Propositions~\ref{prop:equivalent_non_rec} and \ref{prop:equivalent_non_rec_algo}, $\Pi_b$ is equivalent to a subset S of $\Pi'$. Every rule in $S^{-1}=\Pi'\setminus S$ must be contained by S. Since every rule in $\Pi'$ is equivalent to a CQ, every rule in $S^{-1}$ is equivalent to exactly one rule in $S$. It must then be the case that $S^{-1}=\emptyset$ by definition of the use of the $h_{\sf cont}$ membership checking function.
\end{proof}
We illustrate the process with the following examples. 
\begin{example}
    \label{ex:bounded}
    Consider the following program:
    \begin{alignat*}{3}
        & R_1 : & \ r(y) &\ \gets e(x, y).
        \\
        &
        % \quad\quad
        R_2 : &\ r(y) &\ \gets r(x), e(x, y).
    \end{alignat*}
    This program is bounded and is in fact equivalent to one with only a single rule $R_1$.
    Running the algorithm with $\it id$ and $h_{\sf cont}$, we first generate the adorned rule $R_1' : r_{r(y) \gets e(x, y)}(y) \gets e(x, y)$.
    Then, the algorithm would attempt to add $\rho_1: r(y) \gets e(x, y)$ as the adornment of $r(x)$ in the second rule, which then generates the adornment $\rho_2 : r(y) \gets e(x', x), e(x, y)$ for the head atom.
    However, \emph{$\rho_2$ is subsumed by $\rho_1$}, so the algorithm terminates and outputs a program with one rule $R_1'$.
\end{example}
\label{sec:app-trendy-example}
\begin{example}
    We take the following program from \cite[Example~12.5.6]{AHV1995}:
    \begin{align*}
        &R_1 : {\sf Buys}(x,y) \gets {\sf Likes}(x,y). \\
        &
        R_2 : {\sf Buys}(x,y) \gets {\sf Trendy}(x), {\sf Buys}(\_, y).
    \end{align*}
    This program is bounded.
    After running the algorithm, we obtain the following equivalent non-recursive program:
    \begin{align*}
        &R_1' : {\sf Buys}_{{\sf Buys}(x,y) \gets {\sf Likes}(x,y)}(x,y) \gets {\sf Likes}(x,y).\\
        &R_2' : {\sf Buys}_{{\sf Buys}(x,y) \gets {\sf Trendy}(x), {\sf Likes}(\_,y)}(x,y) \gets {\sf Trendy}(x), {\sf Buys}_{{\sf Buys}(x,y) \gets {\sf Likes}(x,y)}(\_, y).
    \end{align*}
    If we continue and attempt to add the adornment $\rho_2 : {\sf Buys}(x,y) \gets {\sf Trendy}(x), {\sf Likes}(\_,y)$ of the head atom of $R_2'$ as an adornment of ${\sf Buys}(\_, y)$, we generate the adornment $\rho_3 : {\sf Buys}(x, y) \gets {\sf Trendy}(x), {\sf Trendy}(x'), {\sf Likes}(\_, y)$. However, $\rho_3$ is subsumed by $\rho_2$, so the algorithm terminates.
\end{example}
From \cref{prop:bounded_eq_terminate}, we have that, even though boundedness is undecidable,
\begin{corollary}
    Bounded datalog programs are recursively enumerable.
\end{corollary}

We note that this does not have to be shown using our framework of adornment. The essence is to perform inlining (which generates rules with only EDB atoms in the body) and eliminate rules that are subsumed by another rule in the flavor of the membership checking function $h_{\sf cont}$.

Nevertheless, the algorithm does not terminate on an unbounded datalog program. To address this, we propose relaxing the procedure using the adornment relaxation function $g_k$ (instead of $id$). This function introduces a budget $k$, retaining all EDB atoms in the adornment (like the identity function does) until there are more than $k$, at which point it starts to behave as $\galgo$. This guarantees termination for all programs but loses the guarantee that all bounded datalog programs can be rewritten into non-recursive equivalents when a rule in the output program requires more than $k$ EDB atoms in an adornment. However, one can choose a concrete value of $k$ based on the time and space budget allocated for the procedure. If the program is unbounded, or if the equivalent non-recursive program is too large (exceeding the specified budget), the algorithm gracefully degrades to use $\galgo$ instead, as it is either impossible or practically unbeneficial to find the non-recursive equivalent. Otherwise, the algorithm will successfully rewrite all bounded datalog programs that fit within the budget k into non-recursive forms. We illustrate this idea by the following example.

\begin{example}
    We take $g_2$ ($g_k$ with $k = 2$) as the adornment relaxation function and consider again the program in \cref{ex:bounded}, but with $R_1$ replaced by $r(x) \gets b(x),$ where $b$ is an EDB relation. This program is no longer bounded. The algorithm starts by generating two EDB-bounded rules:
    \begin{alignat*}{3}
        & R_1' : r_{r(x)\gets b(x)}(x) \gets b(x). \\
        % \quad\quad
        & R_2' : r_{r(y)\gets b(x),e(x, y)}(y) \gets r_{r(x)\gets b(x)}(x), e(x, y).
    \end{alignat*}
    Then, when the algorithm attemps to add $r(y)\gets b(x),e(x, y)$ as the adornment of $r(x)$ in the body of $R_2$, it generates the adornment $r(y) \gets b(x_0), e(x_0, x), e(x, y)$, which now has more than 2 body atoms. Therefore, the adornment relaxation function would now behave as $\galgo$ and transform it to $r(y) \gets e(\_, y)$. We generate the following two additional rules before the procedure terminates:
    \begin{align*}
        & R_3' : r_{r(y) \gets e(\_, y)}(y) \gets r_{r(y)\gets b(x),e(x,y)}(x), e(x, y).\\
        & R_4' : r_{r(y) \gets e(\_, y)}(y) \gets r_{r(y) \gets e(\_, y)}(x), e(x, y).
    \end{align*}
\end{example}

\section{Towards FPT Upper Bounds on the Sizes of IDB Relations}
\label{sec:UpperBoundIdbs}

As a second use case of our framework, we begin our investigation on upper bounds on the size of the IDB relations 
produced by $\Pi$ over an arbitrary (extensional) database instance $D$.
For a given IDB predicate $p$ in $\Pi$, we aim at computing an edge-cover width $w$ in terms of EDB predicates. More specifically, if we have an upper bound $N$ on the size of a single EDB relation in $D$ and the size $|\Pi|$ of the program $\Pi$, then we want to conclude that the size of the IDB relation $p$ resulting from the evaluation of $\Pi$ on $D$ is bounded by $f(\Pi) \cdot N^w$. In other words, over datalog programs with bounded widths, the output size is fixed-parameter tractable with respect to the size of the program.
For this purpose, we again use the adornment relaxation function $\galgo$ 
and membership checking function $h_{\sf eq}$, which allowed us to establish the equivalence between 
an original datalog program $\Pi$ and the EDB-bounded program $\Pi'$ resulting from Algorithm~\ref{alg:adornment}
in Theorem~\ref{thm:equivalent}.

\subsection{Edge-Cover Width}

Towards our goal of providing upper bounds on sizes of IDB relations via adornments, we use a notion analogous to edge cover numbers for conjunctive queries, which is naturally related to size bounds.
We start with the following simple example of an EDB-bounded datalog program.

\begin{example} \label{ex:edge-cover}
    Consider the datalog program $\Pi$ with the following rules:
    \begin{alignat*}{2}
        p(x, y, z) & \gets e(x, y, w), e(x, z, w), e(y, z, w).
        \\
        % \quad\quad
        q(x, y) & \gets p(x, y, \_),
    \end{alignat*}
    where $e$ is an EDB relation. $\Pi$ is equivalent to the following EDB-bounded datalog program $\Pi'$:
    \begin{align*}
        & p_{p(x, y, z) \gets e(x, y, \_), e(x, z, \_), e(y, z, \_)}(x, y, z) \gets e(x, y, w), e(x, z, w), e(y, z, w), \\
        & q_{q(x, y) \gets e(x, y, \_)}(x, y) \gets p_{p(x, y, z) \gets e(x, y, \_), e(x, z, \_), e(y, z, \_)}(x, y, \_).
    \end{align*}
    We note that this EDB-bounded datalog program $\Pi'$ has the following property: if  $(a,b)$ is a tuple in the IDB relation $q$, then there exists some $c$, such that the EDB relation $e$ contains the tuple $(a,b,c)$. In other words, the adornment yields a restriction on the values that can possibly occur in an IDB tuple, and it also yields an upper bound on the size of the IDB relation. In this simple example, the size of the IDB relation $q_\rho$ is clearly bounded by the size of the EDB relation $e$.
\end{example}

The upper bounds on the sizes of IDB relations is captured by the notion of ``edge-cover width'' from a hypergraph, which we define next.

\begin{definition}[hypergraph associated with a rule]
For every rule $q(\bar s) \gets p_1(\bar t_1), \cdots, p_\mu(\bar t_\mu)$, analogously to the case for conjunctive queries,
we define a hypergraph associated with the rule as $H = (V(H), E(H))$, where
the set of vertices $V(H) = \bigcup_{i = 1}^{\mu} \left(\vars(\bar t_i)\right)$, i.e., the set of variables,
and
the set of hyperedges $E(H) = \{ \vars(\bar t_i) \mid 1 \leq 1 \leq \mu \}$.
We let $V_{\sf out} = \vars(\bar s)$ denote the set of output variables, i.e., the set of variables in the head atom $q(\bar s)$.
\end{definition}
    
In the case of a conjunctive query $q$ with associated hypergraph $H = (V(H), E(H))$ and output variables $V_{\rm out}$, the simplest size bound is known to be given by (integral or fractional) edge covers of 
$V_{\sf out}$ using the hyperedges in $E(H)$, formally defined as follows.
\begin{definition}
    An \emph{(integral/fractional) edge cover} of $H = (V(H),E(H))$ with output attributes $V_{\rm out}$ is any assignment $(x_e)_{e \in E(H)}$ that is a feasible solution to the following linear program:
    \begin{alignat*}{2}
        \kappa
        \;=\;
        \min \quad & \sum_{e \in E(H)} x_e \\
        \text{s.t.} \quad
        & \sum_{e \in E(H) :\, v \in e} x_e \;\ge\; 1
        && \qquad \forall v \in V_{\sf out}, \\
        & x_e \in \{0,1\}
    && \qquad \forall e \in E(H), \text{ for integral edge cover, or} \\
    & 0 \le x_e \le 1 
    && \qquad \forall e \in E(H), \text{ for fractional edge cover.}
    \end{alignat*}
    $\kappa$ is called the \emph{(integral/fractional) edge cover number}.

    A \emph{full (integral/fractional) edge cover} of $H = (V(H), E(H))$ is an (integral/fractional) edge cover of $H$ with output attributes $V_{\sf out} = V(H)$.
    
\end{definition}

With the (integral or fractional) edge cover, we can further calculate a size bound of $q$:
\begin{alignat*}{2}
    |q|
    \;\leq\;
    \min \quad & \prod_{e \in E(H)} |e|^{x_e} \\
    \text{s.t.} \quad
    & \sum_{e \in E(H) :\, v \in e} x_e \;\ge\; 1
    && \qquad \forall v \in V_{\sf out}, \\
    & x_e \in \{0,1\}
    && \qquad \forall e \in E(H), \text{ for integral edge cover, or} \\
    & 0 \le x_e \le 1 
    && \qquad \forall e \in E(H), \text{ for fractional edge cover.}
\end{alignat*}
Assuming all EDB relations have at most $N$ tuples, the number of tuples in the final result is bounded by $N^{\kappa}$, where $\kappa$ is the edge cover number.

For general datalog programs, we define an analogous measure called ``edge-cover width'' based on the adornments, as follows.

\begin{definition}[edge-cover width]
    Let $q_\rho(\bar{s})$ be an IDB relation in an EDB-bounded datalog program $\Pi'$.
    Let $H_\rho$ be the hypergraph associated with $\rho$ and $V_{\sf out}$ be the set of output variables of $\rho$.
    Then, the (integral or fractional) \emph{edge-cover width} $\ew(\rho)$ of the adornment $\rho$ is defined by the edge cover number $\kappa$ of $H_\rho$ to cover the output variables $V_{\sf out}$.
    
    The edge-cover width $\ew(q_\rho)$ of the adorned predicate $q_\rho$ is defined as
    $ \ew(q_\rho) = \ew(\rho) $. 
    
    Moreover, the \emph{edge-cover width} $\ew(\Pi')$ \emph{of an EDB-bounded datalog program} $\Pi'$ 
    is given by the maximum edge-cover width $\ew(q_\rho)$ among all EDB-adorned IDB relations $q_\rho$ in $\Pi'$.
\end{definition}

Below, we generalize the observations from \cref{ex:edge-cover}. That is, in an EDB-bounded program, the edge-cover width allows us to establish an upper bound on the sizes of IDB relations and a restriction on the values that can possibly occur in tuples of an IDB.

\begin{theorem}
    \label{thm:single-adorned-size}
    Let $\Pi$ be an EDB-bounded datalog program, and suppose that we evaluate $\Pi$ over a database instance $D$ 
    of size $|D| = N$. Let $k = \ew(q_\rho)$ denote the fractional edge cover width 
    of $q_\rho$. Then every IDB relation $q_\rho$ has size
    $ | q_\rho | \leq  N^{k}. $
\end{theorem}

The above theorem is an immediate consequence of the famous AGM bound~\cite{atserias_size_2013}. 
For the following restrictions on the values that may possibly occur in an IDB tuple, we
have to switch from the fractional to the integral version of $\ew$.

\begin{proposition}
\label{prop:cover-of-qrho}
Let  $\Pi$ be a datalog program, let $\Pi'$ be an EDB-bounded program equivalent to $\Pi$, and let 
$q_\rho$ be an IDB predicate in $\Pi'$,
and let $k = \ew(q_\rho)$ denote the integral edge-cover width.
Then there exist $k$ EDB relations $ e_1, \dots e_k $ with the following property:
for every tuple $\bar s$ in the IDB relation $q_\rho$, 
there exist tuples $\bar t_1 \in e_1, \dots, \bar t_k \in e_k$, such that  
every value $s_\ell$ in $\bar{s}$ occurs in one of the tuples $\bar t_1, \dots, \bar t_k$, i.e., 
there exist $i,j$ with $s_\ell = t_{ij}$.
\end{proposition}

\begin{proof}
    Let $q_\rho$ be an IDB predicate in $\Pi'$ and let 
$S = \{e_1, \dots, e_k\}$ be an integral edge cover of $\rho$.
We are assuming that $\Pi$ is an EDB-bounded datalog program. In particular, this means that 
the adornment $\rho$ bounds all rules with head predicate $q_\rho$, i.e., 
every tuple in the relation $q_\rho$ derived by program $\Pi$ over a database $D$
is also derived by  applying $\rho$ to $D$. Let $\rho'$ denote the rule obtained from $\rho$
by deleting all body atoms that are not in $S$. Clearly, every tuple 
derived by applying $\rho$ to $D$ is also derived by applying $\rho'$ to $D$.
In other words, the IDB relation $q_\rho$ is
obtained by joining the EDB relations $e_1, \dots, e_k$ followed by projection.
In particular, this means that every tuple $\bar s$ in the IDB relation $q_\rho$ gets its 
values from tuples $\bar t_1 \in e_1, \dots, \bar t_k \in e_k$, i.e., 
for every value $s_n$ in $\bar{s}$, 
there exist $i,j$ with $s_n = t_{ij}$.
\end{proof}

By taking the equivalence of a datalog program with an EDB-bounded program into account, 
we immediately get a further corollary.

\begin{proposition}
\label{prop:cover-of-q}
Let  $\Pi$ be a datalog program and let $\Pi'$ be an EDB-bounded program equivalent to $\Pi$. Moreover, let $q$ be an IDB predicate in $\Pi$
and let $k = \ew(q)$ denote the integral edge-cover width.
Then, for every tuple $\bar{s}$ in the IDB relation $q$,  there exist $k$ tuples $\bar{t}_1, \dots, \bar{t}_k$ 
in EDB relations $e_1, \dots, e_k$, respectively, such that every value $s_\ell$ in $\bar{s}$
occurs in some tuple $\bar{t}_i \in e_i$, i.e., there exist $i,j$, 
s.t.\ $t_{ij} = s_\ell$.
\end{proposition}
\begin{proof}
    Let $\bar{s}$ be an arbitrary tuple in the IDB relation $q$. By the equivalence of $\Pi$ and the EDB-bounded program $\Pi'$, 
there exists an adornment $\rho$, such that $\bar{s}$ is also a tuple in the IDB relation $q_\rho$. 
By the definition of the edge-cover width, we have  $k = \ew(q) \geq \ew(q_\rho)$.
By \cref{prop:cover-of-qrho}, there exist (at most) $k$ EDB relations $e_1, \dots, e_k$, such that 
every value $s_n$ in $\bar s$ occurs in some tuple $t_i \in e_i$ with $i \in \{1, \dots, k\}$.
\end{proof}

We have just seen that the edge-cover width 
$\ew(q_\rho)$
of the IDB-predicate $q_\rho$ in an EDB-bounded program $\Pi'$ 
allows us to bound the size of the 
IDB relation $q_\rho$ and to restrict the possible values occurring in $q_\rho$
when evaluating $\Pi'$ over a given EDB $D$. Hence, to make the bound as tight as possible, 
it would be desirable to get $\ew(q_\rho)$ as small as possible. 
Below we show that the EDB-adorned program $\Pi'$ according to Algorithm~\ref{alg:adornment}
is optimal in this sense, i.e., it minimizes  $\ew(q)$ for every IDB predicate $q$ of the original datalog program $\Pi$.

\begin{theorem}
\label{thm:minWidth}
Let  $\Pi$ be a datalog program and let $\Pi'$  
be the EDB-adorned program $\Pi'$ produced by  
Algorithm~\ref{alg:adornment}, when executed with relaxation function $\galgo$
and membership checking function $h_{\sf eq}$.
Then, for every IDB predicate $q$ in $\Pi$, the integral edge-cover width
$\ew(q)$ for the program $\Pi'$ is minimal over all EDB-bounded programs 
equivalent to $\Pi$.
More precisely, let $\Pi''$ be another EDB-bounded program equivalent to $\Pi$. Moreover, let  
$\ew_{\Pi'}(q)$ and $\ew_{\Pi''}(q)$ denote the integral 
edge-cover width of $q$ in 
$\Pi'$ and $\Pi''$, respectively.
Then $\ew_{\Pi'}(q) \leq \ew_{\Pi''}(q)$ holds.  
\end{theorem}

\begin{proof}[Proof sketch]
    If \cref{alg:adornment} gives an edge-cover width $\ew_{\Pi'}(q) = k$ for a predicate $q$, 
    we show a way to construct an EDB instance such that, in some output tuple, there are $k$ pairwise distinct 
    attribute values that have to be covered by $k$ different tuples from this EDB instance.
\end{proof}

\cref{thm:minWidth} allows us to refer to the 
integral edge-cover width of the 
EDB-bounded datalog program $\Pi'$
constructed by \cref{alg:adornment} 
with functions $\galgo$ and $h_{\sf eq}$ as the integral {\em edge-cover width of program} $\Pi$.
We note that, in Theorem~\ref{thm:minWidth}, it is crucial that we define equivalence between the original program $\Pi$ 
and the 
adorned EDB-bounded program $\Pi'$ in terms of {\em all} IDB predicates in $\Pi$ and $\Pi'$. 
If we were only interested in the behavior of $\Pi$ and $\Pi'$ on some distinguished {\em output predicate} $q$,
the edge-cover width of the EDB-bounded program produced by 
\cref{alg:adornment} is not guaranteed to be minimal over all datalog programs
which agree with $\Pi$ on $q$ but may possibly use a completely different
set of additional IDB predicates. 
This is illustrated by the next example. Of course, by the undecidability of 
datalog equivalence~\cite{DBLP:journals/jlp/Shmueli93}, we cannot expect to avoid such situations. 

\begin{example} \label{ex:inline}
    Consider the following datalog program $\Pi$:
    \begin{align*}
        p(w, x, y, z) & \gets e(w), e(x), e(y), e(z),  \\
        % \quad\quad
        q(w) & \gets p(w, \_, \_, \_).
    \end{align*}
    and assume that $q$ is the distinguished output predicate.
    Applying \cref{alg:adornment} with functions  $\galgo$ and  $h_{\sf eq}$, we obtain an EDB-bounded datalog program $\Pi'$ with the following rules:
    \begin{align*}
        & {p}_{p(w, x, y, z) \gets e(w), e(x), e(y), e(z)}(w, x, y, z) \gets e(w), e(x), e(y), e(z), \\
        & {q}_{q(w) \gets e(w)}(w) \gets {p}_{p(w) \gets e(w), e(x), e(y), e(z)}(w, \_, \_, \_).
    \end{align*}
    The adorned IDB relation ${p}_{p(w, x, y, z) \gets e(w), e(x), e(y), e(z)}$ and thus the program $\Pi'$ have $\ew(\Pi')=4$. However, the edge-cover width of the output predicate ${q}_{q(w) \gets e(w)}$ is only $1$. 
    If we perform a simple inlining of the first rule of $\Pi$ into the second before applying \cref{alg:adornment}, 
    we obtain a program $\Pi''$ with only one rule
    \( q_{q(w) \gets e(w)}(w) \gets e(w), \)  
    and hence get edge-cover width $\ew(\Pi'') = 1$.

    This example illustrates that the width defined on the overall datalog program \textit{with a dedicated output predicate} $q$ may change if inlining is employed to eliminate an auxiliary IDB relation with a large edge-cover width (in this case, $p$). Nevertheless, the edge-cover width of $q$ remains invariant and the algorithm still correctly returns the minimum in either case.
\end{example}

As a final remark to conclude this subsection, notice that in the 
Propositions \ref{prop:cover-of-qrho} and \ref{prop:cover-of-q} as well as in 
\cref{thm:minWidth}, it was crucial to take the {\em integral}
edge-cover width. From now on, unless explicitly specified otherwise,
``edge cover'' will always refer to the fractional one.

\subsection{Size Bounds}

We have seen that, for each IDB relation $q$ in some datalog program $\Pi$, in the equivalent EDB-bounded datalog program $\Pi'$ constructed by \cref{alg:adornment}, there are some EDB-bounded IDB relations in the form of $q_\rho$ in $\Pi'$.
In \cref{thm:single-adorned-size}, we have shown that, for each $q_\rho$ in $\Pi'$, $|q_\rho| \leq N^{\ew(q_\rho)}$.
\cref{thm:equivalent} further suggests that the tuples in $q$ will be exactly the union of the tuples in all $q_\rho$.
Therefore, we immediately get the following form of upper bound on the size of $q$:

\begin{proposition}
    For an IDB relation $q$ in a datalog program $\Pi$, let $\Pi'$ be the equivalent EDB-bounded datalog program 
    constructed by \cref{alg:adornment} with functions $\galgo$ and  $h_{\sf eq}$. 
    Then,
    $|q| \leq f(\Pi) \cdot N^{\ew(q)},$
    where $f(\Pi)$ is the number of IDB relations of the form $q_\rho$ in $\Pi'$ (or, the number of different adornments $\rho$ that $q$ takes in $\Pi'$).
\end{proposition}

Note that \cref{alg:adornment} does not rely on any property of the data in the EDB. Therefore, the coefficient $f(\Pi)$ depends solely on the schema of the original program $\Pi$.  But, how can we determine this coefficient, i.e., the number of adornments each IDB relation $q$ in $\Pi$ can take?  We first show that, due to the complexity of this task, it is unlikely that we can get an exact closed-form formula for this number without explicitly generating $\Pi'$:

\begin{theorem}
    \label{thm:sharp-p-hard}
    Given an IDB relation $q$ in a datalog program $\Pi$, the problem of counting the number of adorned relations in the form of $q_\rho$ in the equivalent EDB-bounded datalog program $\Pi'$ generated by \cref{alg:adornment} with $\galgo$ and $h_{\sf eq}$, denoted by \numad, is \SP-hard.
\end{theorem}

\begin{proof}[Proof sketch]
    We use a reduction from the problem of counting the satisfying valuations of a positive 2-CNF, which is known to be \SP-hard~\cite{provan_complexity_1983}. We create one EDB relation per variable.
    Then, for each disjunction $x \vee y$, we create three rules that have the same head $p(\bar s)$ but will have adornments corresponding to (1) $x$, (2) $y$, and (3) both $x$ and $y$, respectively. Finally, we create a rule with head $q(\bar s)$ that has a conjunction of all $p(\bar s)$ as the body. Each adornment of $q$ corresponds precisely to one satisfying assignment of variables in the formula.
\end{proof}

Nevertheless, we can derive an upper bound of $f(\Pi)$ as follows:
\begin{theorem}
    \label{thm:num-adornments-naive}
    For an IDB relation $q$ in a datalog program $\Pi$, let $\Pi'$ be the equivalent EDB-bounded datalog program
    constructed by \cref{alg:adornment} with functions $\galgo$ and  $h_{\sf eq}$. Let $\mathsf{\#EDBs}$ denote the number of EDB predicates and let $\mathsf{ear}$ denote the maximum arity among the set of EDB predicates.
    Then, $f(\Pi)$, the number of IDB relations of the form $q_\rho$ in $\Pi'$, is bounded by
    \[ f(\Pi) \leq (\mathrm{ar}(q) + |\Pi|)^{\mathrm{ar}(q)} \cdot 2^{\#\mathsf{EDBs} \cdot ((\mathsf{ar}(q) + 1)^{\mathsf{ear}} - 1)}. \]
\end{theorem}

\begin{proof}[Proof sketch]
    The factor $(\mathrm{ar}(q) + |\Pi|)^{\mathrm{ar}(q)}$ counts the number of possible head atoms:
    The relation is fixed to be $q$, and each attribute can be either a variable or a constant.
    There are at most $\mathrm{ar}(q)$ different variables,
    and the number of possible constants is bounded by the size of the program, $|\Pi|$.
    Then, we count all subsets of the set of EDB atoms that can possibly be part of the body of the adornment.
    There are $\#\mathsf{EDBs}$ possible relations, and there are at most $\mathsf{ear}$ attributes per atom.
    For each attribute, we choose either one of the $\mathsf{ar}(q)$ variables from the head, or the underscore, and we eliminate the one case where all attributes are underscores.
\end{proof}

Note that this bound overcounts, as we miss an additional requirement that the EDB atoms in the adornment cover all variables in the head atom. However, if we use a different argument, we can get better bounds on the size of an IDB $q$. We start by introducing, as a crucial tool for this proof, the Stirling partition number.

\begin{definition}[Stirling partition number]
    The \emph{Stirling partition number}, or \emph{Stirling number of the second kind}, denoted by $S(n, k)$, counts the number of ways to partition a set of $n$ \emph{labeled} elements into $k$ \emph{non-empty, unlabeled} sets. It can be calculated by
    $S(n, k) = \frac{1}{k!} \sum_{i = 0}^k (-1)^{k - i} \binom{k}{i} i^n,$
    based on the inclusion-exclusion principle.
\end{definition}
This concept aligns with the idea of edge covers for the IDB attributes. We then give the following much tighter bounds.
\begin{theorem}
    \label{thm:size-bounds}
    For an IDB relation $q$ in a datalog program $\Pi$, let $\Pi'$ be the equivalent EDB-bounded datalog program constructed by \cref{alg:adornment} with $\galgo$ and $h_{\sf eq}$. Then,
    \begin{align*}
        |q| & \leq \sum_{k = 1}^{\ew(q)} S(\mathsf{ar}(q), k) \cdot P(\#\mathsf{EDBs} \cdot N, k) \cdot \mathsf{ear}^{\mathsf{ar}(q)} \tag{1} \\
        & \leq ( \#\mathsf{EDBs} \cdot \mathsf{ear} \cdot \mathsf{ar}(q) )^{\mathsf{ar}(q)} \cdot N^{\ew(q)}, \tag{2}
    \end{align*}
    where $P(n, k)$ is the permutation number, and $S(n, k)$ is the Stirling partition number.
\end{theorem}
\begin{proof}[Proof sketch]
    For bound (1), we first enumerate $k$, the number of partitions of a tuple.
    For each $k$, $S(\mathsf{ar}(q), k)$ gives the number of possible ways to partition, and $P(\#\mathsf{EDBs} \cdot N, k)$ gives the number of ways to assign one distinct EDB tuple (out of at most $\#\mathsf{EDBs} \cdot N$) to each of the $k$ partitions.
    Then, each of the $\mathsf{ar}(q)$ attributes can take one of the $\mathsf{ear}$ values in the EDB tuple that is assigned to the partition the attribute is in.
    (1) could be relaxed to obtain (2).
\end{proof}

Note that, in bound (2), the coefficient before $N^{\ew(q)}$ depends only on the schema of the original program $\Pi$. This gives a tighter upper bound
$ f(\Pi) \leq ( \#\mathsf{EDBs} \cdot \mathsf{ear} \cdot \mathsf{ar}(q) )^{\mathsf{ar}(q)}. $

Below, we show multiple examples, in increasing order of generality, to illustrate the correctness and tightness of the bounds.

\begin{example}
    \label{ex:tight}
    \label{ex:width-2-ar-4-to-3}
    \label{ex:4-power-3}
    Consider the following datalog program. We have an EDB relation $e$ with tuples $e(1, 2)$ and $e(3, 4)$. And we have the following rules:
    \begin{align*}
    &R_1 : q(x, y, z) \gets e(x, y), e(z, w) \quad\quad R_2 : q(x, y, z) \gets e(x, y), e(w, z)\\
    &R_3 : q(x, y, z) \gets e(x, w), e(y, z) \quad\quad R_4 : q(x, y, z) \gets e(w, x), e(y, z)\\
    &R_5 : q(x, y, z) \gets q(y, z, x) \quad\quad R_6 : q(x, y, z) \gets q(y, x, z) \quad\quad R_7 : q(x, x, z) \gets q(x, y, z)
    \end{align*}
    This program has only $1$ EDB relation $e$ with arity $2$ and one IDB relation $q$ with arity $3$. The edge-cover width of the program is $2$.
    In this particular case, $q$ will have all tuples $(x, y, z)$ where each of the three values $x, y, z$ can take any value from the domain $\{ 1, 2, 3, 4 \}$---Rules $R_1$ to $R_4$ takes three arbitrary values, and rules $R_5$ to $R_7$ generate all permutations that also allow duplicates.
    So the size of $q$ is $4^3 = 64$.
    The bound (1) gives $|q| \leq 2^3 \cdot \sum_{k = 1}^{2} S(3, k) \cdot P(1 \cdot 2, k) = 8 \cdot (1 \times 2 + 3 \times 2) = 64$,
    which is exactly tight here.
    Nevertheless, if we use the relaxed version (2), we get
    $|q| \leq (\#\mathsf{EDBs} \cdot \mathsf{ear} \cdot \mathsf{ar}(q))^{\mathsf{ar}(q)} \cdot N^{\ew(q)} = (1 \times 2 \times 3)^3 \times 2^2 = 864,$
    which is still correct but overcounts.
\end{example}

\begin{example}
    \label{ex:tight-3-tuples}
    We slightly extend \cref{ex:4-power-3}. We keep the datalog program as-is, but the EDB relation $e$ now has three tuples: $e(1, 2)$, $e(3, 4)$, and $e(5, 6)$.
    How many tuples does $q$ generate then?

    We start with cases where all the three attributes $(x, y, z)$ come from the same tuple in $e$.
    This is the case when, e.g., we generate $q(1, 2, 1)$ from $e(1, 2)$ and $e(1, 2)$ using the rule $R_1$.
    For each such tuple $(a, b)$ in $e$, it can be verified that $q$ will indeed contain all tuples $(x, y, z)$ such that each value $x, y, z \in \{ a, b \}$.
    Since no value overlaps in the EDB relation $e$, we have a total of $3 \times 2^3 = 24$ tuples formed this way.

    Now we continue to the tuples $(x, y, z)$ in $q$ where there are actually two different tuples in $e$ that contribute some value.
    There will be one tuple $(a, b)$ in $e$ that gives two values and one tuple $(c, d)$ that gives one value to this tuple in $q$.
    So we
    (1) choose the tuple $(a, b)$ that gives two values (3 choices),
    (2) choose the tuple $(c, d)$ that gives one value (2 choices),
    (3) choose one of the three positions in the tuple $(x, y, z)$ in $q$ that has a value coming from $(c, d)$ (3 choices, $x$ or $y$ or $z$),
    (4) choose the value to fill in the position we chose in step (3) (2 choices, $c$ or $d$), and finally
    (5) choose the values of the remaining two positions ($2 \times 2 = 4$ choices, as each can take either $a$ or $b$, and duplicates are allowed).
    So we have $3 \times 2 \times 3 \times 2 \times 4 = 144$ tuples this way.

    In total, we get $24 + 144 = 168$ tuples in $q$ in this program.

    Now, if we use the size bound (1):
    \begin{align*}
        |q| & \leq 2^3 \cdot \sum_{k = 1}^{2} S(3, k) \cdot P(1 \times 3, k) = 8 (S(3,1) \cdot P(3, 1) + S(3, 2) \cdot P(3, 2)) = 8 \times (1 \times 3 + 3 \times 6) = 168,
    \end{align*}
    which again is exactly tight.

    If we use the relaxed version (2), we get
    \[ |q| \leq (\#\mathsf{EDBs} \cdot \mathsf{ear} \cdot \mathsf{ar}(q))^{\mathsf{ar}(q)} \cdot N^{\ew(q)} = (1 \times 2 \times 3)^3 \times 3^2 = 1944, \]
    which again is correct but overcounts.
\end{example}

\begin{example}
    \label{ex:tight-general}
    We can further generalize this example. In fact, given any edge-cover width $\ew(q) = \omega$, IDB arity $\mathsf{iar} = \mu$, EDB arity $\mathsf{ear} = \nu$, number of EDB relations $\#\mathsf{EDBs} = m$, and maximum data size $N$ in each EDB relation, we can generate a datalog program that matches exactly the size bound (1), which shows its tightness.
    
    We start by creating $m$ EDB relations $e_1, \dots, e_m$, each with arity $\nu$.
    In each EDB relation $e_i$, we add $N$ data tuples
    \begin{align*}
        (& N\nu(i-1) + 1, && N\nu(i-1) + 2, && \dots, && N\nu(i-1) + \nu), \\
        (& N\nu(i-1) + \nu + 1, && N\nu(i-1) + \nu + 2, && \dots, && N\nu(i-1) + 2\nu), \\
        \vdots \\
        (& N\nu(i-1) + (N-1)\nu + 1, && N\nu(i-1) + (N-1)\nu + 2, && \dots, && N\nu),
    \end{align*}
    where each value occurs in exactly one attribute of exactly one tuple in exactly one EDB relation. There is no duplicate value.
    For example, in $e_1$, we have $(1, 2, \dots, \nu), (\nu+1, \nu+2, \dots, 2\nu), \dots, ((N-1)\nu + 1, (N-1)\nu + 2, \dots, N\nu)$. In $e_2$, we start with $(N\nu + 1, N\nu + 2, \dots, (N+1)\nu)$, etc.
    
    Then, for each $k = 1, \dots, \omega$, for each $(e_{i_1}, \dots, e_{i_k})$ where each $e_{i_\ell} \in \{ e_1, \dots, e_m \}$, for each $(x_{j_1}, \dots, x_{j_\mu})$ where each $x_{j_\ell} \in \{ x_1, \dots, x_{k\nu} \}$, we add a rule
    \[ q(x_{j_1}, \dots, x_{j_\mu}) \gets e_{i_1}(x_1, \dots, x_\nu), e_{i_2}(x_{\nu + 1}, \dots, x_{2\nu}), \dots, e_{i_k}(x_{(k-1)\nu + 1}, \dots, x_{k\nu}). \]
    In essence, this rule chooses $k$ EDB atoms to form an IDB atom in $q$, using $\mu$ values from these $k$ EDB atoms (duplicates allowed).
    It should be clear that this program has exactly the edge cover width $\omega$.
    Since all attributes in all tuples in all EDB relations are distinct, the number of tuples that can be generated is exactly the bound (1), by the same argument as in the proof of \cref{thm:size-bounds}.
\end{example}

To conclude this section, we would like to highlight that, even though the coefficient $f(\Pi)$ appears to have a large upper bound (in Theorems \ref{thm:num-adornments-naive} and \ref{thm:size-bounds}), we expect the exact $f(\Pi)$ to be small for most practical datalog programs that are not artificially made to be adversarial in this respect (e.g., Examples \ref{ex:tight}--\ref{ex:tight-general}, which contain multiple shifting, swapping, and duplication of variables). It is small for, e.g., many classic examples such as transitive closure (\cref{ex:tc-intro}, for which $f(\Pi) = 2$), \cref{ex:bounded} from \cite[Example~12.5.6]{AHV1995}, and same generation \cite[Section~13.4]{AHV1995}. Furthermore, in real-world scenarios, we very often need to assign types to attributes, and each attribute has a semantic meaning. In such cases, variables have very limited possibilities to shift or swap. Therefore, the actual number of adornments tends to be small, and this theoretical worst-case upper bound seems, in general, a gross over-estimation.

\section{Minimizing EDB-Bounded Datalog Programs}
\label{sec:minimize-edb-bounded}

In the previous section, we have seen that, given a datalog program $\Pi$, if we simply follow \cref{alg:adornment} to construct an equivalent EDB-bounded program, the number of adornments of each predicate $q$ (the coefficient in \cref{thm:num-adornments-naive}) can actually be much larger than the coefficient in \cref{thm:size-bounds}. In this section, we propose a way to \emph{minimize} the EDB-bounded datalog program to obtain an equivalent EDB-bounded datalog program with far fewer adornments and rules, so that the number of adornments per IDB predicate matches the coefficient in \cref{thm:size-bounds}.

We first note that the large number of adornments is partially because the same variable may appear in multiple EDB atoms in the adornment. Therefore, to minimize, we enforce that each variable appears exactly once in the body of the adornment.
\begin{definition}[minimal adornment and minimal EDB-bounded program]
    A \emph{minimal} adornment is an adornment in which every variable $v$ in the head of the rule appears in exactly one atom of the body of the adornment. No atom in the adornment consists solely of underscores. 
    
    A \emph{minimal} EDB-bounded datalog program $\Pi''$ is an EDB-bounded datalog program for which every rule has a minimal adornment.
\end{definition}
Recall the adornment relaxation function $\gmcd$ mentioned in \cref{ex:relax-fns}, which retains a minimal set of atoms required to cover the head of a rule and preserves only one occurrence of each variable in the body. Assuming that, given an adornment, there is a consistent, deterministic way to choose which atom the single occurrence of each variable should be in, we show that applying $\gmcd$ on an EDB-bounded program yields an equivalent program with minimal integral edge-cover width.
\begin{theorem}
    \label{thm:minimization}
    Given an EDB-bounded datalog program $\Pi'$. Let $\Pi''$ be the program constructed by applying $\gmcd$ on every adornment of $\Pi'$, then $\Pi''$ is minimal and equivalent to $\Pi'$. Furthermore, $\Pi''$ remains an EDB-bounded datalog program such that the integral edge-cover width $\ew(\Pi'') = \ew(\Pi')$.
\end{theorem}
\begin{proof}
    By definition of $\gmcd$, every adornment in $\Pi''$ is minimal. 
    It is established that renaming invented predicates results in a new program that contains the original one \cite{10.1145/3622847}. Thus, $\Pi' \subseteq \Pi''$. Similarly to the proof of \cref{thm:equivalent}, restoring the original program $\Pi$ from $\Pi''$ is possible by renaming all predicates to their non-adorned version. Therefore, $\Pi' \subseteq \Pi''\subseteq \Pi$. Since $\Pi'$ is equivalent to $\Pi$, it follows that $\Pi''$ is equivalent to both $\Pi'$ and $\Pi$.
    Finally, by construction, $\Pi''$ is an EDB-bounded datalog program, since every head adornment is further relaxed to an adornment that still subsumes the rule. The \emph{integral} edge-cover width of the output program cannot change from $\ew(\Pi')$ since $\gmcd$ only keeps the minimal sets of EDB atoms needed for the cover.
\end{proof}

\begin{example} \label{ex:minimal_edb_bounded}
    The following is a minimal EDB-bounded program equivalent to $\Pi$ in Example~\ref{ex:edge-cover}. It has the same minimal integral edge-cover width of 2.
    \begin{align*}
        & p_{p(x, y, z) \gets e(x, y, \_), e(\_, z, \_)}(x, y, z) \gets e(x, y, w), e(x, z, w), e(y, z, w), \\
        & q_{q(x, y) \gets e(x, y, \_)
        }(x, y) \gets p_{p(x, y, z) \gets e(x, y, \_), e(\_, z, \_)}(x, y, \_). 
    \end{align*}
    Note, however, that this process of removing duplicate occurrences may change the \emph{fractional} edge-cover width. In this example, the adornments in \cref{ex:edge-cover} gives a fractional edge-cover width of 1.5, whereas the minimal program above has a fractional edge-cover width of 2.
\end{example}

We then give the following upper bound of the number of adornments in a minimal EDB-bounded program, which now matches the coefficient $f(\Pi)$ in bound (2) of \cref{thm:size-bounds}.
\begin{theorem}
    \label{thm:minimal-bound}
    Given a datalog program $\Pi$, let $\Pi''$ be the minimal equivalent EDB-bounded datalog program obtained by applying \cref{alg:adornment} with $\galgo$ and $h_{\sf eq}$ followed by minimization with $g_{\sf min}$. Then, for each IDB relation $q$ in $\Pi$, the number of IDB relations of the form $q_\rho$ in $\Pi''$ is bounded by
    \[ \sum_{k=1}^{\ew(q)}(S(\ar(q), k) \cdot \NEDBs^k \cdot \ear^{\ar(q)}) \leq ( \#\mathsf{EDBs} \cdot \mathsf{ear} \cdot \mathsf{ar}(q) )^{\mathsf{ar}(q) } = f(\Pi). \]
\end{theorem}

\section{FPT Complexity Bounds of Evaluation}
\label{sec:complexity-bounds}

In this section, we use EDB adornments and the resulting size bounds to provide fixed-parameter tractable complexity bounds for datalog evaluation.
In addition to the size bounds that we have shown above, we also need some width notions from hypertree decompositions, as we will define next, to reason about the complexity of evaluation.

We start by defining the hypergraphs associated with datalog rules, which are analogous to the case for conjunctive queries as in~\cite{gottlob_hypertree_2002}.

\begin{definition}[hypergraphs associated with datalog rules]
    Given a datalog program $\Pi$, for each rule
    $ R : {q}(\bar{s}) \gets {p_1}(\bar{t}_{1}), \dots, {p_\mu}(\bar{t}_\mu), $
    we define the hypergraph $H$ associated with the rule $R$ as one with hyperedges $E(H) = \{ \vars(\bar t_i) \mid 1 \leq i \leq \mu \}$ and vertices $V(H) = \cup E(H)$.
\end{definition}

The complexity bounds that will follow are based on results of incremental view maintenance, which fundamentally relies on \emph{free-connex} tree decompositions and related width measures~\cite{wang_change_2023, hu_output-optimal_2025}:

\begin{definition}[free-connex generalized hypertree decompositions]
A \emph{hypertree} (for a hypergraph $H$) is a triple $\mathcal{T} = (T,\chi,\lambda)$, where $T$ is a tree, $\chi : V(T) \to 2^{V(H)}$ is a function s.t.\ $\bigcup_{ p \in V(T) } \chi(p) = V(H)$, and $\lambda : V(T) \to 2^{E(H)}$. We call $\HT$ a \emph{free-connex generalized hypertree decomposition (FCGHD)} of $H$ if
\begin{enumerate}[label=(H\textsubscript{\arabic*}),ref=(H\textsubscript{\arabic*})]
	\item\label{H1}\gdef\Giref{\ref{H1}}   For every $e \in E(H)$ there exists some $p \in V(T)$ with $e \subseteq \chi(p)$.
	\item\label{H2}\gdef\Giiref{\ref{H2}}  For every $v \in V(H)$, the set  $\set{ p \in V(T) \mid v \in \chi(p) }$ induces a connected subtree of $T$.
	\item\label{H3}\gdef\Giiiref{\ref{H3}} For every $p \in V(T)$, we have  $\chi(p) \subseteq \vertices{ \lambda(p) }$.
    \item\label{H4}\gdef\Givref{\ref{H4}} There is a connected subtree $S$ of $\mathcal{T}$ such that $\bigcup_{v \in V(S)} \chi(v) = V_{\sf out}$, i.e., the union of attributes that appear in the $\chi$-labels in $S$ is exactly the output attributes.
\end{enumerate}
\Giiref{} is generally called the \emph{connectedness condition}, and \Givref{} is called the \emph{connex condition}. The set $\chi(p)$ for each tree vertex $p \in V(T)$ is conventionally called a \emph{bag}.
\end{definition}

\begin{definition}[(integral/fractional) free-connex generalized hypertree width]
Given a hypertree decomposition $\HT$,
for each node $p \in V(\HT)$, let $V_{\sf full}(p) = \bigcup_{e \in \lambda(p)} e$, i.e., the set of all variables in the hyperedges in the $\lambda$-label.
We define the hypergraph associated with the node $p$ as $H_p = (V_{\sf full}(p), \lambda(p))$.
The (integral or fractional) \emph{width} $w(\HT)$ of a hypertree $\HT = ( T, \chi, \lambda )$ is the maximum (integral or fractional) edge cover width to cover all variables in $V_{\sf full}(p)$, among all nodes $p \in V(T)$.
The (integral or fractional) \emph{free-connex generalized hypertree width}, denoted by $\fchw(p)$, is the minimal width among all possible free-connex generalized hypertree decompositions $\HT$.
\end{definition}

Using the notion of free-connex generalized hypertree width, as well as the size bounds given by the edge-cover width that we introduced earlier, we have

\begin{theorem}
    \label{thm:complexity-general}
    A datalog program $\Pi$ over a database instance with at most $N$ tuples per EDB relation can be evaluated in time
    $O((f(\Pi))^{\fchw(\Pi)} \cdot |\Pi| \cdot N^{\ew(\Pi) \cdot \fchw(\Pi)}),$
    where $\fchw(\Pi)$ is the free-connex generalized fractional hypertree width.
\end{theorem}

\begin{proof}[Proof sketch]
    We assume a semi-na\"ive-style evaluation on the program $\Pi$ and use techniques of incremental view maintenance to evaluate the program.
    We start by evaluating rules whose body consists of EDB atoms only, simply as conjunctive queries.
    Then, in each following iteration, for each rule, we update the relation in the head atom based on the delta from the previous iteration.
    For this, we view the rule again as a conjunctive query, where the relations have sizes at most $f(\Pi) \cdot N^{\ew(\Pi)}$.
    Using results of incremental view maintenance~\cite{abo_khamis_insert-only_2024,wang_change_2023}, we can obtain the bound by performing change propagations on a free-connex hypertree decomposition, which leads to a constant amortized delay for result enumeration.
\end{proof}

If we restrict the scope and consider only \emph{simple chain datalog programs}, i.e., those where each rule contains at most two body atoms, the free-connex hypertree width is at most $2$, and we therefore obtain the following bound as a corollary.

\begin{corollary}
\label{cor:smp_chain_program}
    A \emph{simple chain datalog program} $\Pi$ over a database instance with at most $N$ tuples per EDB relation can be evaluated in time $O((f(|\Pi|))^2 \cdot |\Pi| \cdot N^{2\ew(\Pi)})$.
\end{corollary}

For \emph{linear} datalog programs, i.e., those where each rule has at most an IDB atom in the body, and the rest are all EDB atoms, each with size at most $N$, we have

\begin{theorem}
    \label{thm:linear}
    A \emph{linear} datalog program $\Pi$ (where each rule has at most one IDB atom in the body) over a database instance with at most $N$ tuples per EDB relation can be evaluated in time
    $O(f(\Pi) \cdot |\Pi| \cdot N^{\ew(\Pi) + \fchw(\Pi) - 1}).$
\end{theorem}

\begin{proof}[Proof sketch]
    The strategy is similar to that of \cref{thm:complexity-general}.
    The difference is that, for linear datalog programs, if each EDB relation has size at most $N$ and each IDB relation has size at most $f(\Pi) \cdot N^{\ew(\Pi)}$, then each intermediate relation materialized for each bag of the FCGHD has size at most $f(\Pi) \cdot N^{\ew(\Pi) + \fchw(\Pi) - 1}$, since there are at least $(\fchw(\Pi) - 1)$ EDB relations each with size at most $N$, and the remaining one is either IDB or EDB, but the size is at most $f(\Pi) \cdot N^{\ew(\Pi)}$.
\end{proof}

In addition, we present the complexity bound for a new class of programs:

\begin{definition}
    A datalog program is called \emph{adornment groundable}, if, for every rule,
    \begin{enumerate}[leftmargin=*]
        \item every EDB atom must either have a variable that is not shared with any other atom and appears in the rule's head, or every variable in the EDB atom must appear in the rule's head; and
        \item each IDB variables must appear in the head or in some EDB atom.
    \end{enumerate}
\end{definition}

This definition implies that, for every rule of the equivalent EDB-bounded program, selecting a grounding for the head adornment implies a unique grounding for the entire rule. This class includes classic examples such as transitive closure and reachability. 
\begin{theorem}
\label{thm:adornment_groundable_cmp}
    An adornment groundable datalog program $\Pi$ over a database instance with at most $N$ tuples per EDB relation can be evaluated in time $O(f(\Pi) \cdot |\Pi|\cdot N^{\ew(\Pi)})$.
\end{theorem}

Due to space constraints, we leave more details about this class to the  \arxivorappendix.

%%%%%%%%%%%%%%%%%%%%%%%%%%%%%%%%%%%%%%%%%%%%%%%%%%%%%%%%%%%%%%%%%%%%

\begin{acks}
Qichen Wang is supported by the Ministry of Education, Singapore, through the Academic Research Fund Tier 1 grant (RS32/25), and by the Nanyang Technological University Startup Grant. Part of this work was completed while Qichen Wang was at EPFL.  Reinhard Pichler is supported by the Vienna Science and Technology Fund (WWTF) [10.47379/ICT2201, 10.47379/VRG18013, 10.47379/NXT22018].  
\end{acks}

\bibliographystyle{ACM-Reference-Format}
\bibliography{sample-base}

\clearpage
\SuspendCounters{TotPages}
\newpage

\appendix
\section{Omitted Details in \cref{sec:UpperBoundIdbs}}
\subsection{Proof of \cref{thm:minWidth}}

In order to prove \cref{thm:minWidth},  we first show that the adornments in $\Pi'$ could also be constructed in an alternative way, which is 
based on the following generalization of proof trees: 

\begin{definition}
\label{def:generlizedProofTree}
Let $\Pi$ be a datalog program and let $q(\bar t)$ be an atom \emph{that may contain variables}. 
Then a {\em generalized proof tree} of $q(\bar t)$ w.r.t.\  program $\Pi$ is defined as follows: 

\smallskip
\noindent
{\em Base case.}
If $q$ is an EDB predicate, then the tree consisting of a single node with label $q(\bar t)$ is a 
generalized proof tree of $q(\bar t)$.

\smallskip
\noindent
{\em Inductive case.}
If $q$ is an IDB predicate, then a generalized proof tree $\mathcal{T}$ of $q(\bar t)$ is of the following form: 
there exists a rule $R$ of the form $q(\bar u) \gets p_1(\bar u_1) \dots,  p_\mu(\bar u_\mu)$ and a substitution $\sigma$,
such that the following conditions hold:
\begin{enumerate}[leftmargin=*]
    \item $q(\bar t) = q(\sigma(\bar u))$;
    \item for every $i \in \{1, \dots, \mu\}$, there exists a proof tree $\mathcal{T}_i$ of $p_i(\sigma(\bar u_i))$;
    \item $\mathcal{T}$ has a new root node $r$ labeled with $q(\bar t)$
    and the child nodes
of $r$ are obtained by turning $\mathcal{T}_1, \dots, \mathcal{T}_\mu$ into subtrees
of $\mathcal{T}$, such that the root nodes of $\mathcal{T}_1, \dots, \mathcal{T}_\mu$ 
become the child nodes of~$r$;
\item the edges from the root to its child nodes are labeled with 
    the rule~$R$.
\end{enumerate}
\end{definition}

\noindent
Note that, apart from allowing variables in the atoms labeling the nodes, 
generalized proof trees differ from standard proof trees in that we do not consider them over a concrete database instance. 
Hence, the only condition on the leaf nodes is that they must have an EDB predicate, rather than requesting that these atoms
have to be contained in the EDB.

\begin{example}
Consider the transitive closure example with following rules:
\begin{align*}
    &R_1: \TC(x,y) \gets e(x,y).\\
    &R_2: \TC(x,z) \gets \TC(x,y),e(y,z).
\end{align*}
\cref{fig:gn_proof_tree_tc} shows one generalized proof tree for it. Note, however, that this is not the only possible generalized proof tree. There is, in fact, one generalized proof tree corresponding to each level of inlining of the IDB predicate $\TC$ in $R_2$.
    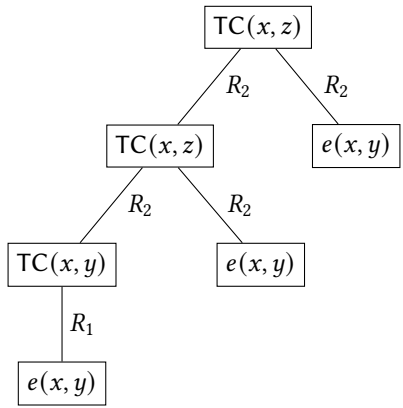
\begin{figure}[h]
    \centering
    \begin{tikzpicture}
\node (e) at (103.75bp,141.75bp) [draw,rectangle] {$\TC(x,z)$};
 \node (d) at (66.75bp,97.50bp) [draw,rectangle] {$\TC(x,z)$};
 \node (f) at (140.75bp,97.50bp) [draw,rectangle] {$e(x,y)$};
 \node (c) at (29.75bp,53.25bp) [draw,rectangle] {$\TC(x,y)$};
 \node (b) at (104.75bp,53.25bp) [draw,rectangle] {$e(x,y)$};
 \node (a) at (29.75bp,9.00bp) [draw,rectangle] {$e(x,y)$};
 \draw [] (e) ..controls (89.996bp,125.17bp) and (80.649bp,114.25bp) .. (d);
 \definecolor{strokecol}{rgb}{0.0,0.0,0.0};
 \pgfsetstrokecolor{strokecol}
 \draw (96.439bp,119.62bp) node {$R_2$};
 \draw [] (e) ..controls (117.5bp,125.17bp) and (126.85bp,114.25bp) .. (f);
 \draw (133.44bp,119.62bp) node {$R_2$};
 \draw [] (d) ..controls (52.996bp,80.92bp) and (43.649bp,70.00bp) .. (c);
 \draw (59.439bp,75.38bp) node {$R_2$};
 \draw [] (d) ..controls (80.876bp,80.92bp) and (90.475bp,70.00bp) .. (b);
 \draw (97.029bp,75.38bp) node {$R_2$};
 \draw [] (c) ..controls (29.75bp,36.67bp) and (29.75bp,25.75bp) .. (a);
 \draw (37.625bp,31.12bp) node {$R_1$};
\end{tikzpicture}
    \caption{Example of one generalized proof tree for transitive closure}
    \label{fig:gn_proof_tree_tc}
\end{figure}
\end{example}

We now present the following lemma on the exact alternative way to construct adornments, namely based on leaf nodes of generalized proof trees.
\begin{lemma}
\label{lem:alternativeAdornedProgram}
Let $\Pi$ be a datalog program and let $\Pi'$ be the EDB-adorned program 
resulting from executing Algorithm~\ref{alg:adornment} with functions $\galgo$ and $h_{\rm eq}$.
Moreover, let $R'$ be a rule in $\Pi'$ with head atom $q_{\rho}(\bar t)$. 
Then $q_{\rho}(\bar t)$ has a generalized proof tree $\mathcal{T}$
with the property that  $\rho$ can be obtained as follows: 
first construct a rule $\rho'$, whose head atom is $q(\bar t)$ and whose body consists of all 
EDB atoms at the leaf nodes of $\mathcal{T}$. Then $\rho$ is obtained from $\rho'$ 
by applying the adornment relaxation function  $\galgo$, i.e., 
replacing all variables outside $\var (\bar t)$ by ``$\_$'',  followed by 
redundancy elimination. 
\end{lemma}

\begin{proof}
As a preparatory step, we introduce the notion of {\em construction depth} $cd(R')$ of a rule $R'$ in $\Pi'$. 
If $R'$ is a rule with only EDB atoms in the body, we set $cd(R') := 1$.
Now suppose that $R'$ is constructed according to the
adornment propagation in Algorithm~\ref{alg:adornment}. That is, $R'$ is constructed from 
some rule $R$ in $\Pi$ together with EDB-adorned rules $R_1, \dots, R_\lambda$ in $\Pi'$ (one for
each IDB atom in the body of $R$). Then we define $cd(R') := 1 + \max \{ cd(R_i) \mid 1 \leq i \leq \lambda \}$.
We prove the lemma by induction on the construction depth of the rule $R'$:

\smallskip
\noindent
{\em Base case.}
Consider a rule $R'$ with $cd(R') = 1$ and with head atom 
$q_{\rho}(\bar t)$.
Suppose that $R'$ is of the form 
$q_\rho(\bar t) \gets p_1(\bar t_1) \dots,  p_\mu(\bar t_\mu)$,
where all body atoms $p_i(\bar t_i)$ are
EDB atoms. 
Then a generalized proof tree $\mathcal{T}$ of 
$q_{\rho}(\bar t)$ is obtained by labeling the root node with 
$q_{\rho}(\bar t)$, appending $\mu$ child nodes to the root with labels
$p_1(\bar t_1) \dots,  p_\mu(\bar t_\mu)$,
and labeling all edges from the root to its child nodes with the rule $R'$.
Moreover, the adornment $\rho$ has head atom $q(\bar t)$ and the body can be obtained in two possible ways:
(1) by \cref{alg:adornment}, we take all body atoms 
of $R'$ followed by an application of $\galgo$; 
(2) by the lemma, we take the EDB atoms labeling the 
leaf nodes of  $\mathcal{T}$  followed by an application of $\galgo$. 
Clearly, both ways lead to the same adornment $\rho$. 

\smallskip
\noindent
{\em Induction step.}
Now consider a rule $R'$ with $cd(R') > 1$ and with head atom 
$q_{\rho}(\bar t)$. By $cd(R') > 1$, $R'$ contains at least one IDB atom in the body. 
W.l.o.g., let $1 \leq \lambda \leq \mu$, such that $p_1, \dots, p_\lambda$ are 
IDB predicates and $p_{\lambda+1}, \dots, p_\mu$ are EDB predicates.
Hence, $R'$ is of the form 
$R'\colon q_\rho(\sigma(\bar s)) \leftarrow 
{p_1}_{\rho_1} (\sigma(\bar s_1)), \dots, {p_\lambda}_{\rho_\lambda} (\sigma(\bar s_\lambda)),
p_{\lambda+1}( \sigma(\bar s_{\lambda+1}))$, $\dots, p_\mu(\sigma(\bar s_\mu))$,
and $q_\rho(\sigma(\bar s)) = q_\rho(\bar t)$.

By the adornment propagation in Algorithm~\ref{alg:adornment},
$R'$ is constructed as follows:  
There exists  a rule  $R\colon q(\bar s) \leftarrow p_1(\bar s_1), \dots, p_\mu(\bar s_\mu)$
in $\Pi$ and rules $R_1, \dots, R_\lambda$ in $\Pi'$, and, w.l.o.g., 
we may assume that $R, R_1, \dots, R_\lambda$ are variable-disjoint.
The head atoms of $R_1, \dots, R_\lambda$ are 
of the form ${p_1}_{\rho_1}(\bar u_1)$, \dots,  ${p_\lambda}_{\rho_\lambda}(\bar u_\lambda)$ and
$\sigma$ is a
simultaneous, most general unifier  
of $\{ (\bar s_1, \bar u_1), \dots, (\bar s_\lambda, \bar u_\lambda) \}$.
We elaborate on the construction of $\rho$ further below
and concentrate on the generalized proof tree first.

Note that all IDB atoms in the body of $R'$ have construction depth $< cd(R')$. 
Hence, by the induction hypothesis, for every $i \in \{1, \dots, \lambda\}$, 
the IDB atom ${p_i}_{\rho_i} (\bar u_i)$ has a generalized proof tree $\mathcal{T}_i$,
s.t.\ $\rho_i$ is obtained as follows:   
first construct a rule $\rho'_i$, whose head atom is ${p_i}_{\rho_i}(\bar u_i)$ and whose body consists of all 
EDB atoms at the leaf nodes of $\mathcal{T}_i$. Then $\rho_i$ is obtained from $\rho'_i$ 
by applying the adornment relaxation function  $\galgo$, i.e.,
replacing all variables outside $\var (\bar u_i)$ by ``$\_$'',  followed by 
redundancy elimination. 

For every $i \in \{1, \dots, \lambda\}$, we can then get a proof tree $\mathcal{T}'_i$ of the instance 
${p_i}_{\rho_i}(\sigma(\bar u_i))$  of ${p_i}_{\rho_i}(\bar u_i)$ by replacing in 
every atom in the tree $\mathcal{T}_i$ every variable $x \in \var(\bar u_i)$ by $\sigma(x)$.
Likewise, we obtain the rule $\sigma(\rho_i)$ by applying the substitution $\sigma$ to the 
head atom ${p_i}_{\rho_i}(\bar u_i)$ and to every  body atom of $\rho_i$.
We can then construct a proof tree $\mathcal{T}$ of $q_\rho(\sigma(\bar s))$
by introducing a new root node $r$ with label $q_\rho(\sigma(\bar s))$ with $\mu$ child nodes and 
labeling the edges to these child nodes by the rule $R'$. The first $\lambda$ child nodes
of $r$ are obtained by turning $\mathcal{T}'_1, \dots, \mathcal{T}'_\lambda$ into subtrees
of $\mathcal{T}$, such that the root nodes of $\mathcal{T}'_1, \dots, \mathcal{T}'_\lambda$ 
become child nodes of $r$. The remaining $\mu - \lambda$ child nodes of $r$ 
are leaf nodes with labels $p_{\lambda+1}( \sigma(\bar s_{\lambda+1}))$, $\dots, p_\mu(\sigma(\bar s_\mu))$.

Finally, according to Algorithm~\ref{alg:adornment},
the adornment $\rho$ of the head atom $q_\rho(\sigma(\bar s))$ of  $R'$ is  obtained by
taking $q(\sigma(\bar s))$ as head atom of $\rho$ and constructing the body of $\rho$
as follows: 
\begin{itemize}[leftmargin=*]
    \item in every EDB atom of the body of $R'$, the variables not occurring in $\sigma(\bar s)$ are replaced by 
    ``$\_$'', and the resulting atom is added to the body of $\rho$;
    \item in every body atom of every rule $\rho_i$, we first apply the substitution $\sigma$ and then again 
    replace the variables not occurring in $\sigma(\bar s)$ by 
    ``$\_$'', before we add the resulting atom to the body of $\rho$.  
    \item we then apply redundancy elimination according to the adornment relaxation function $\galgo$.
\end{itemize}
We now consider an alternative way of constructing the adornment $\rho$
via the generalized proof tree $\mathcal{T}$ constructed above, i.e., 
we follow the construction of $\rho$ according to the lemma.
And we  argue that we thus get exactly the same rule $\rho$ as 
described above via \cref{alg:adornment}.

\smallskip
\noindent
We again take 
$q(\sigma(\bar s))$ as head atom of $\rho$. The body of $\rho$ can be obtained 
as follows:
\begin{itemize}[leftmargin=*]
    \item In every EDB atom of the body of $R'$, the variables not occurring in $\sigma(\bar s)$ are replaced by 
    ``$\_$'', and the resulting atom is added to the body of $\rho$. Note that the EDB atoms of the body of $R'$ 
    are precisely those leaf nodes of $\mathcal{T}$, which are child nodes of the root of $\mathcal{T}$.
    Hence, we can get these body atoms of $\rho$ also by starting off with the leaf nodes 
    immediately below the root of $\mathcal{T}$ and 
    replacing the variables outside $\var (\sigma(\bar s))$ by ``$\_$''.
    \item Now consider the remaining EDB-atoms at the leaf nodes of $\mathcal{T}$, i.e., the leaf nodes of the 
    generalized proof trees $\mathcal{T}_i$. By the induction hypothesis, 
    the body of $\rho_i$ is obtained by taking the EDB atoms at the leaf nodes 
    of $\mathcal{T}_i$, replacing the variables outside $\var(\bar u_i)$ by ``$\_$'', followed by redundancy elimination.
    Apart from the final redundancy elimination step, we thus get the remaining body atoms of $\rho$ by applying 
    $\sigma$ to the body atoms of $\rho_i$ and  
    replacing the variables not occurring in $\sigma(\bar s)$ by ``$\_$''.  Note that first replacing
    variables in an atom $A$ at a leaf node of $\mathcal{T}_i$ outside $\var(\bar u_i)$ by ``$\_$'' to produce atom $A'$, 
    and then 
    replacing in $\sigma(A')$ every variable  outside $\var(\sigma(\bar s))$ by ``$\_$'' produces 
    exactly the same atom as taking atom $A$, applying $\sigma$, and then replacing 
    every variable  outside $\var(\sigma(\bar s))$ by ``$\_$''. The reason for this is that $\sigma$ leaves all variables in the body of $\rho_i$ outside
    $\var(\bar u_i)$ unchanged and the variables in $\rho_i$ outside
    $\var(\bar u_i)$ are also outside $\var(\sigma(\bar s))$. 
    \item We then apply redundancy elimination according to the adornment relaxation function~$\galgo$.
\end{itemize}
It only remains to show that redundancy elimination in every $\rho_i$ followed by redundancy elimination in $\rho$
(referred to as ``two-stage process'' below)
has the same effect as redundancy elimination only applied to $\rho$
(referred to as ``one-stage process'' below). 

We first observe that, in the sequel, 
we may ignore the replacement of some variables by ``$\_$'', since 
it is 
irrelevant for redundancy elimination. That is, 
a body atom $A' = p(\bar u)$ in a rule $\rho$ is redundant, if there exists another body atom 
$B' = p(\bar v)$, such that, for every $i \in \ar(p)$, either $v_i = u_i$ or 
$u_i = \_$. Phrased differently, if $H$ is the head atom of $\rho$, 
then $A'$ is redundant in $\rho$ if there exists another body atom 
$B'$ in $\rho$ with $\vars(A') \cap \vars(H) \subseteq \vars(B') \cap \vars(H)$. 
Note that function $\galgo$ only replaces variables outside $\vars(H)$ by 
``$\_$''. So suppose that $A'$ is obtained from some atom $A$, and
$B'$ is obtained from some atom $B$
by replacing some variables outside $\vars(H)$ by ``$\_$''. 
Then we clearly have 
$\vars(A') \cap \vars(H) \subseteq \vars(B') \cap \vars(H)$
if and only if  
$\vars(A) \cap \vars(H) \subseteq \vars(B) \cap \vars(H)$. 

Hence, if a body atom in $\rho$ constructed by the adornment propagation
in \cref{alg:adornment} or, equivalently, by putting the EDB atoms labeling the 
leaf nodes of $\mathcal{T}$ into the body of $\rho$, then 
redundancy of a body atom of $\rho$ is the same with or without the 
replacement of the variables outside the head of $\rho$ by ``$\_$''. 
Note that this is also true for EDB atoms $A$ labeling a leaf node of one of the
generalized proof trees $\rho_i$ and where the replacement of variables by ``$\_$''
occurs in 2 stages: first replacing variables outside the head of $\rho_i$ and 
then, after applying the mgu $\sigma$, replacing variables outside $\vars(\sigma (\bar s))$
by~``$\_$''.
As has already been argued above, this yields the same result as 
first applying $\sigma$ to $A$ and then replacing variables outside 
$\vars(\sigma (\bar s))$ by~``$\_$''.

We are now ready to prove the redundancy elimination in the one-stage process has the 
same effect as redundancy elimination in the two-stage process. 
We first prove that, if an atom $\sigma(A)$ is redundant in $\rho$ in the one-stage process,
then it is not added to $\rho$ in the two-stage process either.
So let $\sigma(A)$ be redundant in $\rho$ in the one-stage process. 
Then there exists a body atom 
$\sigma(B)$ in $\rho$ with
    $\var(\sigma(A)) \cap \var(\sigma(\bar s)) \subseteq 
    \var(\sigma(B)) \cap \var(\sigma(\bar s))$. 
We have to show that then $\sigma(A)$ is redundant in $\rho$ also in the 
two-stage process. 

First suppose that $\sigma(B)$ is a body atom of $R'$. Then 
$\sigma(B)$ or some atom 
$\sigma(C)$ with
    $\var(\sigma(B)) \cap \var(\sigma(\bar s)) \subseteq 
    \var(\sigma(C)) \cap \var(\sigma(\bar s))$
is  in $\rho$ in the two-stage process. 
In either case, $\sigma(A)$ is also redundant in the two-stage process.
Similarly, suppose that $B$ is the atom at some leaf node of a generalized 
proof tree $\mathcal{T}_i$ and that $B$ is not redundant in $\rho_i$. Then
$\sigma(B)$ or some atom 
$\sigma(C)$ with
    $\var(\sigma(B)) \cap \var(\sigma(\bar s)) \subseteq 
    \var(\sigma(C)) \cap \var(\sigma(\bar s))$
is  in $\rho$ in the two-stage process. 
In either case, $\sigma(A)$ is also redundant in the two-stage process.

It remains to consider the case that 
$B$ is the atom at some leaf node of a generalized 
proof tree $\mathcal{T}_i$ and $B$ is redundant in $\rho_i$.
This means that $\rho_i$ contains an atom $C$ with 
$\var(A) \cap \var(\bar u_i) \subseteq
\var(B) \cap \var(\bar u_i) \subseteq \var(C) \cap \var(\bar u_i)$. 
Then 
$\sigma(C)$ or some atom 
$\sigma(D)$ with
    $\var(\sigma(C)) \cap \var(\sigma(\bar s)) \subseteq 
    \var(\sigma(D)) \cap \var(\sigma(\bar s))$
is  in $\rho$ in the two-stage process. 
That is, we either have
$\var(\sigma(A)) \cap \var(\sigma(\bar s)) \subseteq 
    \var(\sigma(C)) \cap \var(\sigma(\bar s))$ with $\sigma(C)$ in $\rho$ 
    in the two-stage process or 
$\var(\sigma(A)) \cap \var(\sigma(\bar s)) \subseteq 
    \var(\sigma(D)) \cap \var(\sigma(\bar s))$ with 
$\sigma(D)$ in $\rho$ 
    in the two-stage process.
In either case, $\sigma(A)$ is also redundant in the two-stage process.
Note that we have here made use of the following fact: 
let $X, Y$ be atoms at leaf nodes of $\mathcal{T}_i$  with 
$\var(X) \cap \var(\bar u_i) \subseteq
\var(Y) \cap \var(\bar u_i)$. Then 
also 
$\var(\sigma(X)) \cap \var(\sigma(\bar s)) \subseteq 
    \var(\sigma(Y)) \cap \var(\sigma(\bar s))$ holds.
This is due to the fact that only 
variables in $\var(\bar u_i)$ can possibly be mapped by $\sigma$ to 
a variable in $\var(\sigma(\bar s))$.
    
\medskip
\noindent
We now  prove that, if 
an atom $\sigma(A)$ is not added to $\rho$ in the two-stage process,
then it is also redundant in $\rho$ in the one-stage process.
So let $\sigma(A)$ be an atom that is not added to $\rho$ in the two-stage process. 
If $\sigma(A)$ is a body atom in $R'$, then $\rho$ constructed in the 
two-stage process contains an atom $\sigma(B)$ with 
$\var(\sigma(A)) \cap \var(\sigma(\bar s)) \subseteq 
    \var(\sigma(B)) \cap \var(\sigma(\bar s))$.
Note that $\sigma(B)$ cannot be redundant in $\rho$ in the one-stage process, 
since we have shown above that then $\sigma(B)$ would also be redundant in $\rho$ 
in the two-stage process. Hence, $\sigma(B)$ in $\rho$ renders
$\sigma(A)$ redundant also in the one-stage process.

Now consider the case that $A$ is the atom at some leaf node of 
a generalized proof tree $\mathcal{T}_i$ and that $A$ is not redundant in $\rho_i$.
But we are assuming that $\sigma(A)$ is not added to $\rho$. Hence, 
$\rho$ (in the two-stage process) contains an atom $\sigma(B)$ with 
$\var(\sigma(A)) \cap \var(\sigma(\bar s)) \subseteq 
    \var(\sigma(B)) \cap \var(\sigma(\bar s))$.
Again, we may conclude that $\sigma(B)$ cannot be redundant in $\rho$ in the one-stage process
and, therefore, it renders $\sigma(A)$ redundant also in the one-stage process.

It remains to consider the case that 
$A$ is the atom at some leaf node of 
a generalized proof tree $\mathcal{T}_i$ and that $A$ is redundant in $\rho_i$.
Then $\rho_i$ contains an atom $B$ with 
$\var(A) \cap \var(\bar u_i) \subseteq \var(B) \cap \var(\bar u_i)$.
Then $\rho$ (in the two-stage process) either contains $\sigma(B)$ or an atom 
 $\sigma(C)$ with 
$\var(\sigma(B)) \cap \var(\sigma(\bar s)) \subseteq 
    \var(\sigma(C)) \cap \var(\sigma(\bar s))$.
As before,  we conclude that then $\rho$ contains 
either $\sigma(B)$ or $\sigma(C)$ also in the one-stage process, and this atom renders 
$\sigma(A)$ redundant in the one-stage process.   
\end{proof}

The following corollary is an immediate consequence of Lemma~\ref{lem:alternativeAdornedProgram}.

\begin{corollary}
\label{cor:ew-rule-vs-prooftree}    
Let $\Pi$ be a datalog program and let $\Pi'$ be the EDB-bounded program 
resulting from a call of  Algorithm~\ref{alg:adornment} 
with functions $\galgo$ and $h_{\rm eq}$.
Moreover, let $R'$ with head atom $q_\rho(\bar t)$  be a rule in $\Pi'$
and let $S$ be a subset of the body atoms of $\rho$, such that $S$ is a 
minimum-cardinality edge cover of $\var(\bar t)$. Then $S$ is also a minimum-cardinality
edge cover of $\var(\bar t)$ among subsets of the EDB-atoms at the leaf nodes of a 
generalized proof tree of $q_\rho(\bar t)$.
\end{corollary}

\begin{proof}
It has been shown in Lemma~\ref{lem:alternativeAdornedProgram} that, apart from redundancy elimination, the body atoms of 
$\rho$ are precisely the atoms at the leaf nodes of a generalized proof tree $\mathcal{T}$  
of  $q_\rho(\bar t)$. Clearly the replacement of variables outside $\var(\bar t)$ can be 
ignored since these variables never contribute to an edge cover.
It only remains to 
show that an atom $A$ deleted as part of the redundancy elimination can never lead to a smaller edge cover of $q_\rho(\bar t)$. Indeed, suppose that $A$ is part of a cardinality-minimum edge cover of $\var(\bar t)$ from the 
atoms at the leaf nodes of $\mathcal{T}$. The redundancy of $A$ means that there exists an atom $B$
in $\rho$ with $\var(A) \cap \var(\bar t) \subseteq \var(B) \cap \var(\bar t)$. Hence, we can replace $A$ by $B$ and
still have a minimum-cardinality edge cover of $\var(\bar t)$ from the 
atoms at the leaf nodes of $\mathcal{T}$. By iterating this process, we can replace all atoms in the minimum-cardinality
edge cover by atoms from the body of $\rho$. We thus end up with a minimum-cardinality edge cover of $\var(\bar t)$ 
from the atoms  at the leaf nodes of $\mathcal{T}$, which also realizes the edge-cover width of $q_\rho$ 
in~$\rho$.   
\end{proof}

With our notion of generalized proof trees, we can now argue that the construction given by  \cref{alg:adornment}, 
when called with functions $\galgo$ and $h_{\rm eq}$,
indeed produces an EDB-bounded program equivalent to the original program with minimal width, i.e., \cref{thm:minWidth}.

\begin{proof}[Proof of \cref{thm:minWidth}]
Let $q_\rho$ be an adorned IDB predicate in $\Pi'$ that 
realizes the edge-cover width of $q$ in $\Pi'$, i.e.,
$\ew_{\Pi'}(q) = \ew_{\Pi'}(q_\rho)$.
Moreover, let rule $R'$ be a rule in $\Pi'$
with head atom $q_\rho(\bar t)$. 
By Lemma~\ref{lem:alternativeAdornedProgram}, 
$q(\bar t)$ has a generalized proof tree $\mathcal{T}$ in $\Pi'$.

We define the ``canonical database''  of 
$q(\bar t)$ w.r.t.\ $\mathcal{T}$
as the EDB $D$ consisting of the following atoms: 
we start off with the  EDB atoms at the leaf nodes of $\mathcal{T}$, 
but we replace every variable $x_i$ occurring in any atom at the leaf nodes of $\mathcal{T}$
by a distinctive constant $a_i$.
We can then turn the generalized proof tree of $q(\bar t)$ into a (standard) proof tree of
$q_\rho(\sigma(\bar t))$ over the EDB $D$ by applying the substitution $\sigma$ to all variables occurring in 
$\mathcal{T}$. Of course, 
also the rule $\rho$ fires with substitution $\sigma$ over the EDB $D$ 
and produces the IDB atom  
$q_\rho(\sigma(\bar t))$. 

Now let 
$\{x_1, \dots, x_m\}$ denote the variables in $\var(\bar t)$
and let $k = \ew_{\Pi'}(q)$. 
This means that the minimum number of body atoms in $\rho$ and, by 
Corollary~\ref{cor:ew-rule-vs-prooftree}, also the minimum number of atoms from the 
leaf nodes of $\mathcal{T}$ to cover all of $\var(\bar t)$ is $k$.
Likewise, the minimum number of atoms from $D$ to cover all constants
$\{a_1, \dots, a_m\}$ is $k$. 

Now consider another EDB-bounded program $\Pi''$ of $\Pi$. 
We want to show that $\ew_{\Pi''}(q) \geq k$. 
Suppose to the contrary that  $\ew_{\Pi''}(q) < k$. 
By the definition of EDB-boundedness, $\Pi''$ contains a rule 
$R''$ with head predicate $q_{\rho''}$, such that 
$\sigma(\bar t)$ is in the IDB relation $q_{\rho''}$. Hence, 
$\rho''$ fires on the EDB $D$ and produces the 
IDB atom $q(\sigma(\bar t))$. 
But then also the rule $\rho^*$ obtained from $\rho''$ by deleting all body atoms except for 
a minimum-cardinality edge cover fires on $D$ and  produces the 
IDB atom $q(\sigma(\bar t))$. But this means that matching 
the body of $\rho^*$ (which consists of less than $k$ atoms) 
into the EDB $D$ produces less than $k$ atoms of $D$ that 
cover all constants $\{a_1, \dots, a_m\}$. 
This contradicts the above observation that at least $k$ 
atoms from $D$ are required to cover all constants $\{a_1, \dots, a_m\}$.    
\end{proof}

\subsection{Proof of \cref{thm:sharp-p-hard}}

We give a reduction from the following problem which is known to be \SP-hard. (This is the exact problem that is used in the proof in \cite[Section~5]{dalvi_dichotomy_2012}.)
% \begin{theorem}[\cite{provan_complexity_1983,dalvi_dichotomy_2012}]
    Let $X_1, \dots, X_{n_1}$ and $Y_1, \dots, Y_{n_2}$ be two disjoint sets of Boolean variables, and let $G = (V, E)$ be a bipartite graph with vertices $V = \{ X_1, \dots, X_{n_1}, Y_1, \dots, Y_{n_2} \}$ and edges $E \subseteq [n_1] \times [n_2]$. The \emph{positive, partitioned 2-CNF propositional formula (PP2CNF) defined by} $E$ is given by
    \[ \Phi = \bigwedge_{(i, j) \in E} (X_i \vee Y_j). \]
    Then, the problem \pptcnf{}, namely ``given a PP2CNF formula $\Phi$, compute the number of satisfying assignments to the variables $X_1, \dots, X_{n_1}, Y_1, \dots, Y_{n_2}$'', is known to be \SP-hard~\cite{provan_complexity_1983}.
% \end{theorem}
Our goal then is to construct a datalog program $\Pi$ to model such a formula $\Phi$, such that each adornment of some IDB relation maps bijectively to a satisfying assignment.

Assume without loss of generality that the edeges in $E$ are ordered, such that the $k$-th edge in $E$ is $(i, j)$. We construct a datalog program $\Pi$ as follows. We first create one EDB relation $x_i[t]$ for each variable $X_i$ and one EDB relation $y_j[t]$ for each variable $Y_j$. For each edge $(i, j) \in E$ (which has index $k$ in $E$ under our assumption of ordering), we add the following three rules:
\[ (1)~p_k(t) \gets x_i(t), \quad\quad\quad\quad (2)~p_k(t) \gets y_j(t), \quad\quad\quad\quad (3)~p_k(t) \gets x_i(t), y_j(t). \]
And finally we add the rule
\[ q(t) \gets p_1(t), \dots, p_{|E|}(t). \]

Then, in the equivalent EDB-bounded datalog program $\Pi'$ of $\Pi$,
each $p_k$ has three corresponding adorned IDB relations: (1)~${p_k}_{p_k(t) \gets x_i(t)}$, (2)~${p_k}_{p_k(t) \gets y_j(t)}$, and (3)~${p_k}_{p_k(t) \gets x_i(t), y_j(t)}$.
And an adornment $\rho$ of $q_\rho(t)$ has body atoms $B = \{ x_{i_1}(t), \dots, x_{i_a}(t), y_{j_1}(t), \dots, y_{j_b}(t) \}$ for some $1 \leq i_1, \dots, i_a \leq n_1$ and $1 \leq j_1, \dots, j_b \leq n_2$. Let $S = \{ X_{i_1}, \dots, X_{i_a}, Y_{j_1}, \dots, Y_{j_b} \}$, and $A$ be the variable assignment where each $X_i = \mathsf{True}$ if $X_i \in S$, and $X_i = \mathsf{False}$ otherwise; and each $Y_j = \mathsf{True}$ if $Y_j \in S$, and $Y_j = \mathsf{False}$ otherwise. Then,
$\rho$ is a valid adornment of $q(t)$
$\iff$ for each $p_k(t)$, $x_i(t) \in B$ or $y_j(t) \in B$
$\iff$ for each edge $(i, j) \in E$, $X_i \in S$ or $Y_j \in S$
$\iff$ for each edge $(i, j) \in E$, $X_i = \mathsf{True}$ or $Y_j = \mathsf{True}$
$\iff$ for each edge $(i, j) \in E$, $X_i \vee Y_j = \mathsf{True}$
$\iff$ $A$ is a satisfying assignment of $\Phi$.
This shows that the \SP-hard problem \pptcnf{} can be reduced to \numad, so \numad{} is \SP-hard.

\subsection{Proof of \cref{thm:num-adornments-naive}}

    Each adornment $\rho$ is a rule whose head is a predicate of relation $q$ and whose body contains a set of EDB atoms of the form $e(x_1, \dots, x_{\mathsf{ar}(e)})$.
    
    The head atom has $\ar(q)$ arguments, each of which can be either a variable or a constant.
    There can be at most $\ar(q)$ distinct variables (without loss of generality, one can always rename variables so that we use only the same set of $\textrm{ar}(q)$ variable names), and constants come from the datalog program, so there are at most as many constants as the number of terms in the program, which we denote here by $|\Pi|$, the ``size'' of the program.
    
    For each body atom, we first choose one EDB relation---there are $\#\mathsf{EDBs}$ such choices in total.
    Then, we fill in the $\mathsf{ar}(e)$ attributes.
    Each attribute is either one of the (at most $\mathsf{ar}(q)$) variables that appear in the head atom, or an underscore.
    So we have $\mathsf{ar}(q) + 1$ choices for each attribute, and there are $\mathsf{ar}(e)$ attributes in the EDB atom $e(x_1, \dots, x_{\mathsf{ar}(e)})$.
    Excluding one possibility $e(\textunderscore, \dots, \textunderscore)$ where all attributes are underscores, we have a total of $(\mathsf{ar}(q) + 1)^{\mathsf{ar}(e)} - 1$ possible ways to fill in an EDB atom of an EDB relation $e$.
    Then, considering all possible EDB relations, there are at most $\#\mathsf{EDBs} \cdot ((\mathsf{ar}(q) + 1)^{\mathsf{ear}} - 1)$ possible EDB atoms to choose from, where $\mathsf{ear}$ is the maximum arity across all EDBs.
    And finally, there are at most $2^{\#\mathsf{EDBs} \cdot ((\mathsf{ar}(q) + 1)^{\mathsf{ear}} - 1)}$ sets of such EDB atoms.

\subsection{Proof of \cref{thm:size-bounds}}
    \label{sec:app-proof-size-bounds-stirling}

    Consider how tuples in an IDB relation $q$ can be formed.
    By the edge-cover width $\ew(q)$, we know that the values come from at most $\ew(q)$ EDB tuples.
    Therefore, for each tuple in the IDB $q$, we can partition its attributes into at most $\ew(q)$ sets, such that each partition has values coming from one EDB tuple only.
    The number of possible partitions is given by the Stirling partition number (Stirling number of the second kind), $S(\mathsf{ar}(q), k)$, where $k$ ranges from $1$ to $\ew(q)$. Then, for each partition, we choose one EDB tuple to cover it. Assuming that $N$ is the maximum number of tuples in any single EDB relation, we have $(\#\mathsf{EDBs} \cdot N)$ EDB tuples in total, and therefore
    \[ P(\#\mathsf{EDBs} \cdot N, k) = (\#\mathsf{EDBs} \cdot N) \cdot (\#\mathsf{EDBs} \cdot N - 1) \cdot \dots \cdot (\#\mathsf{EDBs} \cdot N - k + 1) \]
    ways to make such choices.
    \begin{example}
        Consider \cref{ex:width-2-ar-4-to-3} again. Assume that there is a tuple $(1, 2, 3)$ in $q$, and we consider where this tuple comes from. We know that 1 and 2 come from the EDB tuple $e(1, 2)$, and 3 comes from the EDB tuple $e(3, 4)$.
        So the way to cover this is to partition $(x, y, z)$ into two parts: $\{x, y\}$ and $\{ z \}$. Then, $\{ x, y \}$ is covered by the tuple $e(1, 2)$, and $\{ z \}$ is covered the tuple $e(3, 4)$.
        If we consider a different tuple, e.g., $(1, 2, 1)$ in $q$, the values all come from the same EDB tuple $e(1, 2)$.
        So all $\{ x, y, z \}$ stay in a single partition, and it is assigned the EDB atom $e(1, 2)$.
    \end{example}
    Then, we sum up $S(\mathsf{ar}(q), k) \cdot P(\#\mathsf{EDBs} \cdot N, k)$ for $k$ ranging from 1 to $\ew(q)$, because we can use at least 1 and at most $\ew(q)$ EDB tuples to cover each tuple in $q$.
    Once we fix a cover, it remains, for each attribute in $q$, to choose one value from the EDB atom that the attribute is assigned to. There are $\mathsf{ear}$ choices for each attribute in $q$ (because there are $\mathsf{ear}$ values in an EDB tuple to choose from), hence $\mathsf{ear}^{\mathsf{ar}(q)}$. So we derive the bound (1).

    However, the Stirling partition number does not have a closed-form formula, except for the one with summation based on the inclusion-exclusion principle.
    If one needs a simpler bound, it is possible to slightly relax it to (2), as follows:
    \begin{align*}
        |q| & \leq \sum_{k = 1}^{\ew(q)} S(\mathsf{ar}(q), k) \cdot P(\#\mathsf{EDBs} \cdot N, k) \cdot \mathsf{ear}^{\mathsf{ar}(q)} \tag{1} \\
        & \leq (\#\mathsf{EDBs} \cdot N)^{\ew(q)} \cdot \mathsf{ear}^{\mathsf{ar}(q)} \cdot \sum_{k = 1}^{\ew(q)}S(\mathsf{ar}(q), k) & (k \leq \ew(q), P(n, k) \leq n^k) \\
        & \leq ( \#\mathsf{EDBs} \cdot \mathsf{ear} )^{\mathsf{ar}(q)} \cdot \left( \sum_{k = 0}^{\mathsf{ar}(q)}S(\mathsf{ar}(q), k) \right) \cdot N^{\ew(q)} & (\ew(q) \leq \mathsf{ar}(q)) \\
        & = ( \#\mathsf{EDBs} \cdot \mathsf{ear} )^{\mathsf{ar}(q)} \cdot B_{\mathsf{ar}(q)} \cdot N^{\ew(q)} \\
        & \leq ( \#\mathsf{EDBs} \cdot \mathsf{ear} \cdot \mathsf{ar}(q) )^{\mathsf{ar}(q)} \cdot N^{\ew(q)}, & (B_n \leq n^n) \tag{2}
    \end{align*}
    where $B_n$ is the $n$-th Bell number, which is similar to the Stirling partition number but counts the total number of possible partitions of all sizes.

\section{Omitted Details in \cref{sec:minimize-edb-bounded}}
\subsection{Proof of \cref{thm:minimal-bound}}

    Each adornment has at most $\ew(q)$ atoms. The head variables can be partitioned into $k$ sets in $S(\mathsf{ar}(q), k)$ different ways. $k$ EDB predicates (repetitions allowed) can be assigned to the adornment in no more than $\NEDBs^k$ ways. Finally, variables can be assigned positions within their respective EDB atoms in $\textsf{ear}^{\ar(q)}$ ways.
    We can then follow the proof of \cref{thm:size-bounds} in \cref{sec:app-proof-size-bounds-stirling}.

\section{Omitted Details in \cref{sec:complexity-bounds}}
\subsection{Proof of \cref{thm:complexity-general}}

    We assume a semi-na\"ive-style evaluation on the program $\Pi$, with techniques of incremental view maintenance.  We first introduce two important theorems for incremental view maintenance when the update sequence is insertion-only, i.e., we will only insert new records into the database without deleting any records from it:

    \begin{corollary}[Derived from {\cite[Theorem~4.1]{abo_khamis_insert-only_2024}}]
        There exists an algorithm $\mathcal{A}_1$, such that given any conjunctive query $Q$, and an initially empty database $D$, $\mathcal{A}_1$ can maintain any insertion-only updates $t$ to the database $D$ in amortized $O(N^{\fchw(Q)-1}\cdot |t|)$ time while supporting constant delay enumeration of the delta query $\Delta Q \gets Q(D+t)-Q(D)$, where $N$ is the size of total updates.
    \end{corollary}

    \begin{corollary}[Derived from {\cite[Theorem~6.10]{wang_change_2023}}]
        There exists an algorithm $\mathcal{A}_2$, such that given any conjunctive query $Q$ with $\fchw(Q) = 1$, and an initially empty database $D$, $\mathcal{A}_2$ can maintain any insertion-only updates $t$ to the database $D$ in amortized $O(|t|)$ time while supporting constant delay enumeration of the delta query $\Delta Q \gets Q(D+t) - Q(D)$, where $N$ is the size of total updates.
    \end{corollary}

    In \cite{abo_khamis_insert-only_2024} and \cite{wang_change_2023}, the authors provide a detailed explanation for designing one possible $\mathcal{A}_1$ and $\mathcal{A}_2$.  In this paper, we consider $\mathcal{A}_1$ and $\mathcal{A}_2$ as two black-box algorithms, without inquiring into any details of the algorithms. We denote $\mathcal{A}_1(t)$ and $\mathcal{A}_2(t)$ as the algorithms that maintain the data structures with insertion $t$ given as input, and enumerate the corresponding delta results.

    Given a datalog program $\Pi$, for every IDB relation $p$, we consider it as a view $V(p)$ to be maintained by $\mathcal{A}_1(p)$ and $\mathcal{A}_2(p)$.  In addition, we consider all EDB relations to be empty for the initial database.  For the first iteration of computation, we insert all EDB relations and evaluate rules $R$ that contain only EDB atoms in the body.
    Each such rule can finish in one iteration with complexity
    $O(N^{\fchw(R)} \cdot (f(\Pi) \cdot N^{\ew(R)})) = O(f(\Pi) \cdot N^{\fchw(\Pi) + \ew(\Pi)})$,
    where $N$ is the maximum size of any single EDB relation (input), and $f(\Pi) \cdot N^{\ew(\Pi)}$ is the size bound of the IDB relation in the head atom (output).
    In the following iterations,
    by maintaining updates, some new tuples $\Delta p$ may be generated for $V(p)$.  After updating all IDB views $V(p)$, if there are some new tuples $\Delta p$ generated for any IDB view $V(p)$, we enumerate $\Delta p$ in constant delay and move on to the next iteration, where we use $\Delta p$ to be the insertion-only update for view $V(p)$, and insert $\Delta p$ into any $\mathcal{A}_1(p')$ and $\mathcal{A}_2(p')$, if $p$ is the body atom of $p'$.  By recursively enumerating and inserting the delta queries, the maintenance procedure stops when no more delta queries can be generated.  

    Now consider the cost of the maintenance procedure for each individual $V(p)$.   Given a $V(p)$, we first construct a free-connex hypertree decomposition for $p$.  For each bag of relations $B$, we maintain the bag using $\mathcal{A}_1$.  By treating each bag as a single relation $B$, we then obtain a query with $\fchw = 1$, which can be maintained using $\mathcal{A}_2$.   For an EDB relation, its size is bounded by $O(N)$, and for an IDB relation, its size is bounded by $O(f(\Pi) \cdot N^{\ew(\Pi)})$.  The output of each bag is thus bounded by $O((f(\Pi) \cdot N^{\ew(\Pi)})^{\fchw(B)})=O((f(\Pi))^{\fchw} \cdot N^{\ew(\Pi) \cdot {\fchw}(\Pi)})$.
    Then, when maintaining the resulting width-1 query, the input size (number of tuples per bag) is at most $(f(\Pi) \cdot N^{\ew(\Pi)})^{\fchw(\Pi)}$, and the output size of the rule is at most $f(\Pi) \cdot N^{\ew(\Pi)}$.  

    Summing over all rules across all iterations, the overall complexity is given by
    \begin{align*}
        & O\bigg(
            \sum_{\textrm{iteration $i$}}\ 
            \sum_{\textrm{rule $R$ in $\Pi$}} \ 
            \sum_{\substack{\text{body atom $p(\bar{t})$ in $R$} \\ \text{$p$ is updated in iteration $i-1$} \\ \text{with delta $\Delta p_{i-1}$}}}
            (O(\mathcal{A}_1(\Delta p_{i-1})) + O(\mathcal{A}_{2}(\Delta p_{i-1})) \bigg) \\
        = & O\bigg( \sum_{\textrm{rule $R$}} \bigg(
            \sum_{\substack{\text{body atom $p(\bar t)$} \\ \text{iteration $i$}}} O(\mathcal{A}_1(\Delta p_{i-1}))
            + \sum_{\substack{\text{body atom $p(\bar t)$} \\ \text{iteration $i$}}} O(\mathcal{A}_{2}(\Delta p_{i-1})) \bigg) \bigg) \\
        = & O\bigg( \sum_{\textrm{rule $R$}} \bigg( (f(\Pi) \cdot N^{\ew(\Pi)})^{\fchw(\Pi)} + (f(\Pi) \cdot N^{\ew(\Pi)})^{\fchw(\Pi)} + f(\Pi) \cdot N^{\ew(\Pi)} \bigg) \bigg) \\
        = & O\bigg( \sum_{\textrm{rule $R$}} (f(\Pi) \cdot N^{\ew(\Pi)})^{\fchw(\Pi)} \bigg) \\
        = & O((f(\Pi))^{\fchw(\Pi)} \cdot |\Pi| \cdot N^{\ew(\Pi) \cdot \fchw(\Pi)}). \qedhere
    \end{align*}

\subsection{Proof of \cref{thm:linear}}

    We again assume semi-na\"ive-style evaluation, with techniques of incremental view maintenance using algorithms $\mathcal{A}_1$ and $\mathcal{A}_2$ as in the proof of \cref{thm:complexity-general}.

    Given a datalog program $\Pi$, for every IDB relation $p$, we consider it as a view $V(p)$ to be maintained by $\mathcal{A}_1(p)$ and $\mathcal{A}_2(p)$.  In addition, we consider all EDB relations to be empty in the initial database.  For the first iteration of computation, we insert all EDB relations and evaluate rules $R$ that contain only EDB atoms in the body.
    Each such rule can finish in one iteration with complexity
    $O(N^{\fchw(R)} \cdot (f(\Pi) \cdot N^{\ew(R)})) = O(f(\Pi) \cdot N^{\fchw(\Pi) + \ew(\Pi)})$,
    where $N$ is the maximum size of any single EDB relation (input), and $f(\Pi) \cdot N^{\ew(\Pi)}$ is the size bound of the IDB relation in the head atom (output).
    By maintaining these updates, some new tuples $\Delta p$ may be generated for $V(p)$.  After updating all IDB views $V(p)$, if there are some new tuples $\Delta p$ generated for any IDB view $V(p)$, we enumerate $\Delta p$ in constant delay and move on to the next iteration, where we use $\Delta p$ to be the insertion-only update for view $V(p)$, and insert $\Delta p$ into any $\mathcal{A}_1(p')$ and $\mathcal{A}_2(p')$, if $p$ is the body atom of $p'$.  By recursively enumerating and inserting the delta queries, the maintenance procedure stops when no more delta queries can be generated.  

    Now consider the cost of the maintenance procedure for each individual $V(p)$.   Given a $V(p)$, we first construct a free-connex hypertree decomposition for $p$.  We maintain each bag using $\mathcal{A}_1$.  By treating each bag as a single relation, we then obtain a query with $\fchw = 1$, which can be maintained using $\mathcal{A}_2$.   For an EDB relation, its size is bounded by $O(N)$, and for an IDB relation, its size is bounded by $O(f(\Pi) \cdot N^{\ew(\Pi)})$.  
    Since $\Pi$ is linear, each bag has at most one IDB relation, so the output of each bag is thus bounded by $f(\Pi) \cdot N^{\ew(\Pi)} \cdot N^{\fchw(\Pi) - 1}$.  
    When maintaining the resulting width-$1$ query, the total input size would be $N$ (for an EDB relation), or $f(\Pi) \cdot N^{\ew(\Pi)} \cdot N^{\fchw(\Pi) - 1}$ (for an IDB relation), and the output of the rule is bounded by $O(f(\Pi) \cdot N^{\ew(\Pi)})$.  

    Summing over all rules across all iterations, the overall complexity is given by
    \begin{align*}
        & O\bigg(
            \sum_{\textrm{iteration $i$}}\ 
            \sum_{\textrm{rule $R$ in $\Pi$}} \ 
            \sum_{\substack{\text{body atom $p(\bar{t})$ in $R$} \\ \text{$p$ is updated in iteration $i-1$} \\ \text{with delta $\Delta p_{i-1}$}}}
            (O(\mathcal{A}_1(\Delta p_{i-1})) + O(\mathcal{A}_{2}(\Delta p_{i-1})) \bigg) \\
        = & O\bigg( \sum_{\textrm{rule $R$}} \bigg(
            \sum_{\substack{\text{body atom $p(\bar t)$} \\ \text{iteration $i$}}} O(\mathcal{A}_1(\Delta p_{i-1}))
            + \sum_{\substack{\text{body atom $p(\bar t)$} \\ \text{iteration $i$}}} O(\mathcal{A}_{2}(\Delta p_{i-1})) \bigg) \bigg) \\
        = & O\bigg( \sum_{\textrm{rule $R$}} \bigg( (f(\Pi) \cdot N^{\ew(\Pi) + \fchw(\Pi) - 1} + f(\Pi) \cdot N^{\ew(\Pi) + \fchw(\Pi) - 1} \bigg) \bigg) \\
        = & O\bigg( \sum_{\textrm{rule $R$}} f(\Pi) \cdot N^{\ew(\Pi) + \fchw(\Pi) - 1} \bigg) \\
        = & O((f(\Pi))^{\fchw(\Pi)} \cdot |\Pi| \cdot N^{\ew(\Pi) \cdot \fchw(\Pi)}). \qedhere
    \end{align*}

\subsection{Complexity Bounds by Grounding}\label{app:HornSat}
We now examine an alternative method of evaluating datalog programs by using a propositional Horn-SAT program. This method allows us to derive the complexity bound for the adornment groundable class of datalog Program.
\begin{definition}[equivalent propositional Horn-SAT program]
    A \emph{propositional Horn program} $\Sigma$ is a set of rules of the form: $$p_0 \gets p_1 \land p_2 \land \dots\land p_n$$ where $p_0, p_1,\dots p_n$ are Boolean variables. The program $\Sigma$ is \emph{equivalent} to a datalog program $\Pi$ if and only if there is a one-to-one mapping between all the derived ground atoms of $\Pi$ and the minimal model of $\Sigma$.
\end{definition}
\begin{definition}
    Grounding is the process of replacing variables in an atom or rule with specific ground terms from the program's domain to create an instance that contains no variables.
\end{definition}
\begin{example}
The following is a valid grounding for the transitive closure program on the database with tuples $e(0,1)$ and  $e(1,2)$:
    \begin{align*}
        &\TC(0, 1) \gets e(0,1).\\
        &\TC(1, 2) \gets e(1,2).\\
        &\TC(0, 2) \gets \TC(0, 1), e(1,2).
    \end{align*}
\end{example}

We now show an example of how grounding the adornment of the head, gives a unique grounding for the entire rule:
\begin{example}
    Consider the EDB-Bounded datalog rule for transitive closure: 
    $$\TC_{\TC(x,z)\gets e(x,y), e(y,z)}(x,z) \gets \TC_\rho(x,y), e(y,z).$$ 
    Upon grounding the head's adornment with $\TC(0,2)\gets e(0,1), e(1,2)$, the grounding for the entire body can be inferred. The grounding for the head is the same as the grounding for the head of its adornment.    
    The EDB atom $e(y,z)$ shares a unique variable, $z$, with the head and must appear in the head adornment. Consequently, it is assigned the grounding $e(1,2)$. The IDB atom $\TC_\rho(x,y)$ is constrained by the head and the grounded EDB atom, leading to the unique grounding $\TC_\rho(0,1)$. Thus we get $\TC_{\TC(0,2)\gets e(0,1), e(1,2)}(0,2) \gets \TC_\rho(0,1), e(1,2)$.
\end{example}

We now give a proof for \cref{thm:adornment_groundable_cmp}:
\begin{proof}
    The datalog program $\Pi$ is analyzed via its EDB-bounded version, $\Pi'$, by transforming it into an equivalent propositional Horn clauses. The clauses can be generated by choosing grounding for the head adornment. Since all adornments can be covered by one of its subsets of size at most $\ew(\Pi)$, there are only $N^{\ew(\Pi)}$ possible groundings per rule. There are $|\Pi'|$ rules in the program, leading to at most $|\Pi'|\cdot N^{\ew(\Pi)}$ Horn clauses. Let $k$ be the maximum number of atoms within a rule. The set of Horn clauses contains at most $(k+1)\cdot f(\Pi)\cdot N^{\ew(\Pi)}$ symbols. Solving a set of Horn clauses can be done in time linear to the number of symbols~\cite{dowling_linear-time_1984}. Thus the overall running time is bounded by $O(|\Pi'|\cdot N^{\ew(\Pi)})$.
    Since $|\Pi'| \leq f(\Pi) \cdot |\Pi|$ by \cref{thm:minimal-bound}, the running time is bounded by $O(f(\Pi) \cdot |\Pi| \cdot N^{\ew(\Pi)})$.
\end{proof}

\end{document}